\RequirePackage{amsmath}
\documentclass[onecolumn,journal]{IEEEtran} 

\usepackage[T1]{fontenc}

\usepackage[l2tabu,orthodox]{nag}

\usepackage[british]{babel}

\usepackage{microtype}
\usepackage{graphicx}
\usepackage{color}
\usepackage{enumitem}
\usepackage[normalem]{ulem}

\usepackage{float}
\usepackage{array}
\usepackage{tabularx}
\newcolumntype{L}[1]{>{\raggedright\arraybackslash}p{#1}}
\newcolumntype{C}[1]{>{\centering\arraybackslash}p{#1}}
\newcolumntype{R}[1]{>{\raggedleft\arraybackslash}p{#1}}
\usepackage{booktabs}

\newcommand{\comment}[1]{}

\newcommand{\Rnum}[1]{\uppercase\expandafter{\romannumeral #1\relax}}
\newcommand{\rnum}[1]{\lowercase\expandafter{\romannumeral #1\relax}}

\usepackage[
    advantage,
    adversary,
    asymptotics,
    complexity,
    events,
    ff,
    lambda,
    landau,
    logic,
    mm,
    notions,
    operators,
    oracles,
    primitives,
    probability,
    sets]{cryptocode}

\renewcommand{\pcthen}{\highlightkeyword{\,:\ }}
\renewcommand{\pcelse}{\highlightkeyword{else:\ }}
\newcommand{\pcbreak}{\highlightkeyword{break}}

\newcommand{\pcand}{\highlightkeyword{and}}

\usepackage{amssymb}
\usepackage{stmaryrd}

\usepackage{amsthm}
\usepackage{amsfonts}
\usepackage{xspace}

\theoremstyle{plain}
\newtheorem{theorem}{Theorem}
\newtheorem{corollary}[theorem]{Corollary}

\newtheorem{lemma}[theorem]{Lemma}
\newtheorem{proposition}[theorem]{Proposition}

\theoremstyle{definition}
\newtheorem{definition}[theorem]{Definition}

\theoremstyle{remark}

\newtheorem{remark}[theorem]{Remark}

\usepackage{hyperref}
\hypersetup{colorlinks=true,linkcolor={red!90!black},citecolor=darkgray}
\usepackage[capitalize,noabbrev]{cleveref}

\AddToHook{env/lemma/begin}{\crefalias{theorem}{lemma}}
\AddToHook{env/proposition/begin}{\crefalias{theorem}{proposition}}
\AddToHook{env/corollary/begin}{\crefalias{theorem}{corollary}}
\AddToHook{env/definition/begin}{\crefalias{theorem}{definition}}
\AddToHook{env/example/begin}{\crefalias{theorem}{example}}
\AddToHook{env/claim/begin}{\crefalias{theorem}{claim}}
\AddToHook{env/remark/begin}{\crefalias{theorem}{remark}}

\crefname{theorem}{Theorem}{Theorems}
\crefname{lemma}{Lemma}{Lemmas}
\crefname{proposition}{Proposition}{Propositions}
\crefname{corollary}{Corollary}{Corollaries}
\crefname{definition}{Definition}{Definitions}
\crefname{example}{Example}{Examples}
\crefname{claim}{Claim}{Claims}
\crefname{remark}{Remark}{Remarks}

\usepackage[color={.9 .2 .2},final]{attachfile2}

\usepackage{tikz}
\usetikzlibrary{cd,shapes}
\usepackage{pgfplots}
\pgfplotsset{compat=1.18}
\usetikzlibrary{pgfplots.fillbetween}

\usepackage[noadjust]{cite}

\usepackage{orcidlink}

\newcommand{\GL}{\mathcal{GL}}

\newcommand{\Ort}{\mathcal{O}}

\newcommand{\from}{\leftarrow}

\newcommand{\iso}{\cong}
\newcommand{\binomq}{\genfrac{[}{]}{0pt}{}}

\DeclareMathOperator{\Gal}{Gal}
\DeclareMathOperator{\Span}{Span}
\DeclareMathOperator{\diag}{diag}
\DeclareMathOperator{\rank}{rank}
\DeclareMathOperator{\hull}{hull}

\newcommand{\variety}{\ensuremath{\mathcal{V}}}

\newcommand{\Unif}[1]{\ensuremath{\mathcal{U}\left( #1 \right)}}

\newcommand{\code}{\ensuremath{\mathfrak{C}}}

\DeclareMathOperator{\PCE}{PCE}

\DeclareMathOperator{\OSPCE}{(S)PCE}
\DeclareMathOperator{\LCE}{LCE}

\renewcommand{\vec}[1]{\ensuremath{\mathbf{#1}}\xspace}
\newcommand{\mat}[1]{\ensuremath{\mathbf{#1}}\xspace}
\newcommand{\inner}[2]{\ensuremath{\left\langle #1, #2 \right\rangle}}
\newcommand{\transpose}{\mathtt{T}}

\newcommand{\algo}[1]{\ensuremath{\mathsf{#1}}\xspace}
\newcommand{\var}[1]{\ensuremath{\mathsf{#1}}\xspace}

\newcommand{\condset}[2]{\ensuremath{\left\{ #1 \; \middle\vert \; #2 \right\}}}
\newcommand{\Mid}{\; \Vert \;}

\newcommand{\Dual}{\algo{Dual}}
\newcommand{\Hull}{\algo{Hull}}
\newcommand{\Dehull}{\algo{Dehull}}
\newcommand{\SampleLCD}{\algo{SampleLCD}}
\newcommand{\FindPlane}{\algo{FindPlane}}
\newcommand{\FindPlaneH}{\algo{FindPlaneH}}
\newcommand{\SampleSelfOrthogonal}{\algo{SSO}}
\newcommand{\SampleCode}{\algo{SampleCode}}
\newcommand{\SampleHollow}{\algo{SampleHollow}}
\newcommand{\SampleTotal}{\algo{SampleTotal}}
\newcommand{\DualGalois}{\algo{DualGalois}}
\newcommand{\HullGalois}{\algo{HullGalois}}
\newcommand{\DehullGalois}{\algo{DehullGalois}}

\mathchardef\mhyphen="2D

\ifCLASSOPTIONtwocolumn
\NewDocumentCommand \iftwocolumns {+m +m} {#1}
\else
\NewDocumentCommand \iftwocolumns {+m +m} {#2}
\fi

\ifCLASSOPTIONpeerreview
\NewDocumentCommand \ifsupplementary {+m +m} {#2}
\else
\NewDocumentCommand \ifsupplementary {+m +m} {#1}
\fi

\begin{document}

\title{Algorithms for Sampling Self-Orthogonal and Totally Self-Orthogonal Codes in Odd Characteristic}

\author{%
Martin~R.~Albrecht$^{\orcidlink{0000-0003-0093-4578}}$%
\thanks{Albrecht is at King's College London and SandboxAQ.},
Benjamin~Ben{\v c}ina$^{\orcidlink{0009-0005-8891-8389}}$%
\thanks{Ben{\v c}ina is at Royal Holloway, University of London, and Aalto University. Corresponding author: \texttt{benjamin.bencina@aalto.fi}.},
Russell~W.~F.~Lai$^{\orcidlink{0000-0001-9126-1887}}$%
\thanks{Lai is at Aalto University.}
\thanks{The research of Albrecht, Ben{\v c}ina and Lai was supported by UKRI Grant EP/Y02432X/1, by UKRI Grant EP/S021817/1, and by the Research Council of Finland grant 358951 and the Future Makers Program of the Technology Industries of Finland Centennial Foundation via the project SPICED, respectively.}
}

\maketitle

\begin{abstract}
  We give an algorithm that samples uniformly random linear codes of any hull dimension and type over finite fields of odd characteristic, implying the first algorithm for sampling uniformly random self-orthogonal codes of any rate, including self-dual codes. Our algorithm is a re-visitation of the algorithm given by Albrecht, Ben{\v c}ina and Lai (EC'25), using the mass formulae proven by Li, Shi and Ling (IEEE Trans.~Inf.~Theory~71(1)). This allows us to instantiate code-based cryptographic schemes that rely on the hardness of Permutation Code Equivalence (PCE) on `random' self-dual codes for security and that were previously unable to sample them. Building on the observation by Bardet, Otmani and Saeed-Taha (ISIT'19) that Euclidean orthogonality is insufficient when considering PCE over finite extension fields due to non-trivial Galois automorphisms, we study the behaviour of what we call total orthogonality, that is orthogonality with respect to all induced Galois geometries simultaneously. We characterise total orthogonality of vectors and codes, and give an algorithm that samples linear codes with a total hull of a prescribed dimension; a subcode that acts as the hull in all Galois geometries of the ambient space. The algorithm incurs a rate decrease by a factor equal to the extension degree, however, we argue why this may be necessary in the context of PCE and explore how it limits the practicality of our algorithm. We consider the notion of Galois type of a linear code when Galois hulls are symmetric and show that all linear codes have constant Hermitian type.
\end{abstract}

\ifCLASSOPTIONpeerreview
\else
\begin{IEEEkeywords}
  Self-orthogonal codes, self-dual codes, Galois geometries, permutation equivalence over finite extension fields.
\end{IEEEkeywords}
\fi

\section{Introduction}\label{sec:intro}

\IEEEPARstart{O}{ne} of the long-standing hardness assumptions in code-based cryptography states that self-orthogonal codes, and self-dual codes in particular, yield the hardest Permutation Code Equivalence \((\PCE)\) instances. This was first studied in~\cite{ITIT:PetRot97}, the connection between the hardness of \(\PCE\) and the dimension of the hull of the code was made in~\cite{ITIT:Sendrier00}, and the so-called \emph{hull attacks} for non-binary codes appeared in~\cite{PQCRYPTO:SenSim13,ISIT:BarOtmSae19}, having complexity exponential in the hull dimension and polynomial in every other parameter. There were even cryptographic constructions that based their security on the hardness of \(\PCE\) on random self-orthogonal or self-dual codes, most notably the \(\PCE\) regime of the original LESS signature scheme~\cite{PQCRYPTO:BBPS21}, whose \(\LCE\) version has been submitted to NIST for standardisation~\cite{NISTPQC-ADD-R2:LESS24}.

However, the question of how to sample from the distribution of uniformly random self-orthogonal codes with given parameters remained largely unanswered until recently. 
While the authors of~\cite{EC:AlbBenLai25} gave the first known algorithm for sampling uniformly random linear codes over odd prime fields with a desired hull dimension, the nature of their sampling algorithm limited them to codes with hulls of co-dimension at least \(2\).\footnote{This was not a problem for them, as they also needed a lattice problem to be hard simultaneously, which depended on the co-dimension of the hull.} A second attempt was made in their follow-up work~\cite{upke2}, which indirectly gives a sampling algorithm for codes over odd prime fields of any hull dimension, including self-orthogonal codes, but again of rate below \(\frac{k}{2k + 2}\) (up to duality). They also show in~\cite{upke2} that \(\PCE\) over finite extension fields is significantly easier than \(\PCE\) over prime fields when other parameters remain the same -- for an extension field of degree \(d\), the Frobenius structure of the extension field allows the adversary to construct \(d\) \(\PCE\) samples from a single sample, allowing them to use a linearisation algorithm~\cite{AC:BCDSK24} to solve \(\PCE\).

\subsection{Our Contributions}\label{sec:intro:contrib}

In this work we give algorithms for sampling self-orthogonal codes of any rate over a field of odd characteristic. For finite extension fields, we define and explore a more appropriate notion of orthogonality, give the corresponding sampling algorithms, and discuss hard \(\PCE\) instances.

\paragraph*{Sampling self-orthogonal and self-dual codes}
Our algorithm is, essentially, the slightly modified algorithm \(\SampleCode\) from \cite[Def.~11]{EC:AlbBenLai25} for sampling close to the uniform distribution on the set of \([n,k]\)-linear codes over \(\FF_q\) with a prescribed hull dimension \(h\). We provide a formal analysis of this algorithm using the mass formulae from~\cite{ITIT:LiShiLin25} that allow us to prove it is correct, and that its output distribution conditioned on not aborting is exactly (not just close to as claimed in~\cite{EC:AlbBenLai25}) the uniform distribution on the set of \([n,k]\)-linear codes over \(\FF_q\) with a prescribed hull dimension \(h\) and \emph{square class} \(\chi \in \set{\pm 1}\). The parameters of our algorithm have no restriction (bar the obvious conditions on the dimension), in particular implying the first known algorithm that samples uniformly from the set of self-dual codes of any admissible dimension over \(\FF_q\) for any odd prime power \(q\). Our algorithm can thus be used to formally instantiate multiple code-based primitives that use `random self-dual codes', such as the LESS signature scheme~\cite{PQCRYPTO:BBPS21} in the \(\PCE\) regime\footnote{The version submitted to NIST~\cite{NISTPQC-ADD-R2:LESS24} only includes the \(\LCE\) regime.} and the recently proposed SPECK signature scheme~\cite{EPRINT:BBMSSZ25}.

\paragraph*{Total orthogonality over finite extension fields}
While our algorithms for sampling self-orthogonal codes work over any field of odd characteristic, the simple notion of Euclidean orthogonality is insufficient for many applications, such as generating plausibly hard \(\PCE\) instances. As noted briefly in~\cite[Sec.~IV.A]{ISIT:BarOtmSae19}, it is necessary to consider other `inner products' besides the Euclidean form, namely the sesquilinear Galois inner products given by non-trivial automorphisms of the extension field, introduced in~\cite{DCC:FanZha17,DCC:FuLiu22}. Indeed, for each \(\phi_i \in \Gal(\FF_q/\FF_p) \iso \ZZ_d\) for \(q = p^d\) and \(p\) prime we consider the sesquilinear forms \(\inner{\vec{x}}{\vec{y}}_{\phi_i} = \inner{\vec{x}}{\phi_i(\vec{y})}\). Each of these forms induces its own dual and hence its own hull, which could be trivial even for Euclidean self-orthogonal codes. After giving algorithms to compute Galois duals and Galois hulls (these are duals and hulls with respect to Galois inner products), we set out to explore \emph{total orthogonality}, that is orthogonality with respect to all degree-many Galois inner products simultaneously. We show in \cref{thm:orthogonal} that two vectors \(\vec{x},\vec{y}\) (and hence any subcodes) are totally orthogonal over \(\FF_{q}\) if and only if the codes \(\Pi(\vec{x}), \Pi(\vec{y})\) obtained by a naturally defined coefficient projection are orthogonal over \(\FF_p\), ignoring the multiplication structure of \(\FF_q\) and only considering it as an \(\FF_p\)-vector space.

\paragraph*{Sampling totally hollow codes}
We are particularly interested in codes with a large \emph{total hull}, that is the intersection of all Galois hulls, giving plausibly hard \(\PCE\) instances over extension fields. Using our above characterisation of total orthogonality, we give in \cref{thm:algo:sampletotal} an algorithm sampling \([n,k]\)-linear codes over \(\FF_{p^d}\) with total hull dimension \(h\) and a fixed square class for any odd prime \(p\), provided there exists an algorithm sampling \([n, d k]\)-linear codes over \(\FF_p\) with hull dimension \(d h\). Instantiating this with our algorithm for sampling codes of any type, we obtain an algorithm for sampling codes of any total hull dimension with rate at most \(\frac{1}{2 d}\) for odd prime \(p\). This rate loss is motivated by the Frobenius sample amplification technique from~\cite{upke2} -- if the rate is not divided by \(d\), there exists a linearisation algorithm akin to~\cite{AC:BCDSK24} solving \(\PCE\) over \(\FF_{p^d}\). In other words, even if there exists an algorithm for sampling totally self-orthogonal codes (or any lower total hull dimension) that circumvents the rate loss, such codes may not be cryptographically useful as they yield easy \(\PCE\) instances.
Combining our rate bound of \(\frac{1}{2 d}\) with the recently proposed algorithm for \(\PCE\) with small rates~\cite{ITIT:BatMorSan26} yields a limit on the hardness of \(\PCE\) over extension fields -- unless \(d \lesssim \sqrt{\frac{n}{8}}\), there exists an efficient heuristic distinguisher for \(\PCE\) on totally hollow codes of rate \(\frac{1}{2 d}\), which is an upper bound up to duality for codes produced by our algorithm.

\subsection{Open Problems}\label{sec:intro:open}

The main open problem in this area is extending our results to the one remaining case, that is (totally) hollow codes over fields of characteristic \(2\), which per \cref{thm:algo:sampletotal,thm:algo:samplecode} reduces to giving an algorithm that samples binary codes of any hull dimension and `type' (cf.~\cite{ITIT:LiShiLin25}). Another possible future direction is quantifying totally hollow codes, that is how many there are for particular parameters, describing the exact output distribution of our algorithm in \cref{def:algo:sampletotal}, and how totally hollow codes relate to equivalence classes in different variants of code equivalence. The Galois square classes in \cref{sec:total:arith} (and by extension Galois types) need to be further explored as they provide another easy avenue to deciding \(\PCE\) akin to~\cite{localinv} -- an instance pair needs to have matching Galois hull dimensions and Galois square classes in order to be plausibly hard.

\section{Preliminaries}\label{sec:prelim}

We write \([\mat{A} \Mid \mat{B}]\) for \([\mat{A}^\transpose \mid \mat{B}^\transpose]^\transpose\), that is stacking \(\mat{A}\) on \(\mat{B}\). Denote \(\ZZ_n = \ZZ/ n\ZZ\) for \(n \in \ZZ\) and \(\llbracket n, m \rrbracket = [n,m] \cap \ZZ\) for \(n, m \in \ZZ\) with \(n \leq m\). For a vector \(\vec{v}\) we denote by \(\vec{v}[i]\) its \(i\)-th entry and by \(\vec{v}\llbracket i, j\rrbracket\) the vector with \(j-i+1\) entries that are the entries of \(\vec{v}\) at indices \(\llbracket i, j \rrbracket\). The same notation for a matrix \(\mat{A}\) denotes the \(i\)-th row and a range of consecutive rows, respectively. We call \(\adv\) a \((t, \varepsilon)\)-algorithm if it terminates in time \(t\) and has success probability \(\varepsilon\) (otherwise it aborts with \(\bot\)). An algorithm is Probabilistic Polynomial Time (\(\ppt\)), or simply efficient, if \(t\) is a polynomial in its inputs. Logarithms have base \(2\).

\subsection{Finite Fields}\label{sec:prelim:fields}

Let \(\FF_q\) for \(q = p^d\) be a finite field of prime characteristic \(p\) and degree \(d\), and \(\FF_q^\star = \FF_q\setminus\set{0}\) its group of units. It is constructed as \(\FF_p[X]/f(X)\) for an irreducible polynomial \(f\) over \(\FF_p\) of degree \(d\), which always exists but is not necessarily unique. A root \(z \in \FF_q\) of \(f(X)\) when viewed over its splitting field \(\FF_q\) yields a basis for \(\FF_q\) as a \(d\)-dimensional vector space over \(\FF_p\) of the form \(\mathcal{B}_z = \set{1, z,\dots,z^{d-1}}\). The polynomial \(f\) determines how this basis behaves under field multiplication in \(\FF_q\), namely how the element \(z^d\) wraps around.

The Galois group \(\Gal(\FF_q/\FF_p) \iso \ZZ_d\) of automorphisms of \(\FF_q/\FF_p\) is generated by the \(p\)-Frobenius map \(\phi \colon x \mapsto x^p\), that is each automorphism of \(\FF_q\) is of the form \(\phi_i \colon x \mapsto x^{p^i}\) for \(i \in \ZZ_d\). The isomorphism with \(\ZZ_d\) simply maps each \(\phi_i\) to \(i\), thus we sometimes denote \(\phi_{-i} = \phi_{d-i} = \phi_i^{-1}\). We extend each automorphism \(\phi\) to vectors \(\vec{x} \in \FF_q^n\) entry-wise, that is \(\phi(\vec{x}) = [\phi(x_i)]_{i=1}^n\), and likewise to matrices.

Denote by \(\chi_q \colon \FF_q^\star \to \set{\pm 1}\) the quadratic character, which is the multiplicative homomorphism \(x \mapsto x^\frac{q-1}{2}\) that maps quadratic residues to \(1\) and non-residues to \(-1\).

We call the affine conic \(F(x_1, x_2) = a x_{1}^2 + b x_{1} x_{2} + c x_{2}^2 + d x_{1} + e x_{2} + f\) over \(\FF_q\) for odd \(q\) \emph{smooth}, or non-degenerate, if its discriminant \(D = b^2 - 4 a c\) is non-zero. We have the following statement on the number of points.

\begin{proposition}[{\cite{BDR:conic}}]\label{prop:conic}
  Let \(F(x_1,x_2) = a x_1^2 + b x_1x_2 + c x_2^2 + d x_1 + e x_2 + f\) with \(D = b^2 - 4 a c \neq 0\) be a smooth affine conic over \(\FF_q\) for \(q\) odd prime power. Then
  \[
    \left\lvert\variety_F\right\rvert = \begin{cases}
      1       & \text{if } f' = 0     \text{ and } D \in \chi_q^{-1}(-1) \\
      q - 1   & \text{if } f' \neq 0  \text{ and } D \in \chi_q^{-1}(1)  \\
      q + 1   & \text{if } f' \neq 0  \text{ and } D \in \chi_q^{-1}(-1) \\
      2 q - 1 & \text{if } f' = 0     \text{ and } D \in \chi_q^{-1}(1),
    \end{cases}
  \]
  where \(f' = f + \frac{c d^2 - b d e + a e^2}{D}\).
\end{proposition}

\subsection{Linear Codes}\label{sec:prelim:codes}

An \([n,k]\)-linear code over \(\FF_q\) is an \(\FF_q\)-vector subspace \(\code \leq \FF_q^n\) of dimension \(k\) and length \(n\). A code is represented by a \(k \times n\) full-rank row-basis matrix \(\mat{A}\) over \(\FF_q\), denoted here as \(\code(\mat{A})\), which is unique only up to multiplication from the left with a basis change matrix \(\mat{U} \in \GL_k(\FF_q)\). We normally take the canonical representative that is row-reduced. Traditionally, a linear code \(\code\) (especially over a prime field) is equipped with the Euclidean inner product defined as the bilinear \(\inner{\vec{x}}{\vec{y}} = \sum_{i=1}^{n} x_i y_i\). Its dual \(\code^\perp\) is defined as its orthogonal complement, that is \(\code^\perp = \condset{\vec{y} \in \FF_q^n}{\forall \vec{x} \in \code \colon \inner{\vec{x}}{\vec{y}} = 0}\). Calling it a complement is misleading, however, since the intersection \(\code \cap \code^\perp\) may be non-trivial. This intersection is a subspace of \(\code\) and is called the \emph{hull} of \(\code\)~\cite{DM:AssKey90}, and has been subject to extensive investigation~\cite{SIAM:Sendrier97,ITIT:Sendrier00,PQCRYPTO:SenSim13,ISIT:BarOtmSae19} due to its important role in determining whether two codes are permutation equivalent. For completeness, we include the algorithms for computing the Euclidean dual and hull in \cref{fig:algo:dual,fig:algo:hull}. A vector \(\vec{x}\) is called self-orthogonal if \(\inner{\vec{x}}{\vec{x}} = 0\), a code is called self-orthogonal if \(\code \subseteq \code^\perp\), and self-dual if \(\code = \code^\perp\). Recall self-dual codes have \(n = 2 k\) and exist if and only if \(\chi_q(-1)^k = 1\) (that is \(q = 1 \pmod{4}\) or \(k\) even). Codes with a trivial hull for which \(\code \cap \code^\perp = \set{0}\) holds, and so \(\code \oplus \code^\perp = \FF_q^n\), are called \emph{Linear Complementary Dual} (LCD) codes~\cite{DM:Massey92}. More generally, a code with hull dimension \(h\) will be called \(h\)-hollow as in~\cite{EC:AlbBenLai25}. There is an efficient algorithm given in~\cite{localinv} for `dehulling' a linear code, that is finding a matrix \(\mat{T} \in \FF_q^{(k-h) \times n}\) called an LCD component that modulo \(\hull(\code)\) spans \(\code/\hull(\code)\). For completeness we include it in \cref{fig:algo:dehull}. The square class \(\chi_\code\) of \(\code\), introduced in~\cite{localinv}, is \(\chi_q(\det(\mat{T}\mat{T}^\transpose))\) for any LCD component of \(\code\) and is, essentially, the discriminant of any Gram matrix of \(\code\) modulo the hull. The pair \((\dim(\hull(\code)), \chi_\code)\) is called its type (cf.~\cite{ITIT:LiShiLin25}). Let \(\mathcal{X}_{n,k,q}\), \(\mathcal{X}_{n,k,q}^h\), and \(\mathcal{X}_{n,k,q}^{\chi,h}\) denote the set of all \([n,k]\)-linear codes over \(\FF_q\), those with respect to hull dimension \(h\), and additionally with respect to square class \(\chi\), respectively. The uniform distribution on the set of generator matrices of codes in \(\mathcal{X}_{n,k,q}\), \(\mathcal{X}_{n,k,q}^h\) and \(\mathcal{X}_{n,k,q}^{\chi,h}\) is denoted by \(\mathsf{U}_{n,k,q}\), \(\dist_{n,k,q}^h\) and \(\dist_{n,k,q}^{\chi,h}\), respectively. All self-orthogonal codes have square class \(1\), and we have \(\chi_{\code^\perp} = \chi_{\code} \chi_q(-1)^h\)~\cite[Prop.~5]{ITIT:LiShiLin25}. The square class and hull dimension can be seen as arithmetic invariants by assigning to a code generator matrix \(\mat{A} \in \FF_q^{k \times n}\) the quadratic space \(V_\mat{A} = (\FF_q^k, \mat{A}\mat{A}^\transpose)\) and computing its Witt decomposition, as done in~\cite{localinv}. Types of codes then correspond (up to dimension parity) to types of finite geometry~\cite[Sec.~III.6]{Artin1957}.

\begin{figure}[t]
  \centering
  \begin{pcvstack}[boxed]
    \procedure[mode=text]{$\Dual(\mat{A})$}{
      \pcif \(\mat{A} = [\cdot] \in \FF_q^{0 \times n}\)\pcthen \pcreturn \(\mat{I}_n\) \\
      \(\mat{A}^\perp \coloneqq\) basis of right kernel of \(\mat{A}\);\;
      \pcreturn \(\mat{A}^\perp\)
    }
    \pcvspace
    \procedure[mode=text]{$\Hull(\mat{A})$}{
      \(\mat{A}^\perp \coloneqq \Dual(\mat{A})\) \\
      \(\mat{H} \coloneqq\) basis of right kernel of \([\mat{A} \Mid \mat{A}^\perp]\);\;
      \pcreturn \(\mat{H}\)
    }
    \pcvspace
    \procedure[mode=text]{$\Dehull(\mat{A})$}{
      \(\mat{H} \coloneqq \Hull(\mat{A})\);\;
      \(\mat{V} \coloneqq\) basis of \(\code(\mat{A}) / \code(\mat{H})\) \\
      \(\mat{T} \coloneqq\) lift \(\mat{V}\) to \(\code(\mat{A})\);\;
      \pcreturn \((\mat{T}, \mat{H})\)
    }
  \end{pcvstack}
  \caption{The algorithms \(\Dual\), \(\Hull\), and \(\Dehull\). Given a code generator matrix \(\mat{A} \in \FF_q^{k \times n}\), they compute the dual, the hull, and a `dehulling' of \(\code(\mat{A})\), respectively. All algorithms are deterministic if they return row-reduced bases.}
  \label{fig:algo:dual}\label{fig:algo:hull}\label{fig:algo:dehull}
\end{figure}

The first progress since~\cite{SIAM:Sendrier97} in obtaining a usable and precise mass formula for linear codes with a particular hull dimension other than for self-orthogonal codes was made in~\cite{ITIT:CMTQ19} for the case of LCD codes.

\begin{proposition}[{\cite[Thm.~25, Thm.~31 \& Cor.~32]{ITIT:CMTQ19}}]\label{prop:count:lcd}
  Let \(n,k,q \in \ZZ\) with \(0 \leq k \leq n\) and \(q\) odd prime power. Denote by \(\mathcal{X}_{n,k,q}^{+1,0}\) (resp.~\(\mathcal{X}_{n,k,q}^{-1,0}\)) the set of all \([n,k]\)-linear LCD codes over \(\FF_q\) with a basis whose Gram matrix is \(\diag(1,\dots,1)\) (resp.~\(\diag(1,\dots,1,s)\) for some quadratic non-residue \(s \in \FF_q^\star\)), and as before by \(\mathcal{X}_{n,k,q}^0\) the set of all \([n,k]\)-linear LCD codes over \(\FF_q\). Then \(\mathcal{X}_{n,k,q}^0 = \mathcal{X}_{n,k,q}^{+1,0} \sqcup \mathcal{X}_{n,k,q}^{-1,0}\), and we have
  \[
    \left\lvert \mathcal{X}_{n,k,q}^{\pm1,0} \right\rvert =
    \begin{cases}
      \frac{1}{2} q^\frac{k \cdot (n-k) - 1}{2} A \binomq{\frac{n}{2}-1}{\frac{k-1}{2}}_{q^2} & \text{if \(k\) odd and \(n\) even}, \\
      \frac{1}{2} q^\frac{k \cdot (n-k)}{2} B_\pm \binomq{\frac{n-1}{2}}{\frac{k-1}{2}}_{q^2} & \text{if \(k\) odd and \(n\) odd}, \\
      \frac{1}{2} q^\frac{k \cdot (n-k)}{2} C_\pm \binomq{\frac{n-1}{2}}{\frac{k}{2}}_{q^2} & \text{if \(k\) even and \(n\) odd}, \\
      \frac{1}{2} q^\frac{k \cdot (n-k)}{2} D_\pm \binomq{\frac{n}{2}}{\frac{k}{2}}_{q^2} & \text{if \(k\) even and \(n\) even},
    \end{cases}
  \]
  where \(A = q^\frac{n}{2} - \chi_q(-1)^\frac{n}{2}\), \(B_\pm = q^\frac{n-k}{2} \pm \chi_q(-1)^\frac{n-k}{2}\), \(C_\pm = q^\frac{k}{2} \pm \chi_q(-1)^\frac{k}{2}\), \(D_\pm = \frac{\left( q^\frac{k}{2} \pm \chi_q(-1)^\frac{k}{2} \right) \cdot \left( q^\frac{n-k}{2} \pm \chi_q(-1)^\frac{n-k}{2} \right)}{\left( q^\frac{n}{2} + \chi_q(-1)^\frac{n}{2} \right)}\), hence
  \[
    \left\lvert \mathcal{X}_{n,k,q}^0 \right\rvert =
    \begin{cases}
      q^\frac{k \cdot (n-k) - 1}{2} A \binomq{\frac{n}{2}-1}{\frac{k-1}{2}}_{q^2} & \text{if \(k\) odd and \(n\) even}, \\
      q^\frac{(k+1) \cdot (n-k)}{2} \binomq{\frac{n-1}{2}}{\frac{k-1}{2}}_{q^2} & \text{if \(k\) odd and \(n\) odd}, \\
      q^\frac{k \cdot (n-k+1)}{2} \binomq{\frac{n-1}{2}}{\frac{k}{2}}_{q^2} & \text{if \(k\) even and \(n\) odd}, \\
      q^\frac{k \cdot (n-k)}{2} \binomq{\frac{n}{2}}{\frac{k}{2}}_{q^2} & \text{if \(k\) even and \(n\) even}.
    \end{cases}
  \]
\end{proposition}

This was recently extended in~\cite{ITIT:LiShiLin25} to codes of any hull dimension, a concurrent work to~\cite{EC:AlbBenLai25}. It uses a similar technique to our procedure for sampling uniformly random codes with a prescribed hull dimension which we present in \cref{thm:algo:samplecode} to instead count them. Using this, they recursively extend the result of~\cite{ITIT:CMTQ19} into the following general statement.

\begin{proposition}[{\cite[Thm.~5 \& 6]{ITIT:LiShiLin25}}]\label{prop:count:hollow}
  Let \(n,k,h,q \in \ZZ\) with \(0 \leq h \leq k \leq n-h\) and \(q\) odd prime power. Denote by \(\mathcal{X}_{n,k,q}^{+1,h}\) (resp.~\(\mathcal{X}_{n,k,q}^{-1,h}\)) the set of all \([n,k]\)-linear codes over \(\FF_q\) with hull dimension \(h\) and a basis whose Gram matrix is \(\diag(1,\dots,1,0,\dots,0)\) (resp.~\(\diag(1,\dots,1,s,0,\dots,0)\) for some quadratic non-residue \(s \in \FF_q^\star\)), the number of \(0\)s being \(h\), and as before by \(\mathcal{X}_{n,k,q}^{h}\) the set of all \([n,k]\)-linear codes over \(\FF_q\) with hull dimension \(h\). Denote also \(t = k-h\). Then \(\mathcal{X}_{n,k,q}^{h} = \mathcal{X}_{n,k,q}^{+1,h} \sqcup \mathcal{X}_{n,k,q}^{-1,h}\) (where \(\mathcal{X}_{n,k,q}^{-1,k}\) is empty when \(h=k\)), and we have
  \[
    \left\lvert \mathcal{X}_{n,k,q}^{\pm1,h} \right\rvert =
    \begin{cases}
      C_1 \left\lvert \mathcal{X}_{n,t,q}^{\pm1,0} \right\rvert & \text{if \(n-t\) odd}, \\
      C_2 A_{\pm} \left\lvert \mathcal{X}_{n,t,q}^{\pm1,0} \right\rvert & \text{if \(n-t\) even},
    \end{cases}
  \]
  where \(A_\pm = \frac{q^{\frac{n-t}{2}-h} \pm \chi_q(-1)^\frac{n-t}{2}}{q^\frac{n-t}{2} \pm \chi_q(-1)^\frac{n-t}{2}}\) and \(C_i = \prod_{i=1}^{h} \frac{q^{n-t-2i+\ell} - 1}{q^{t} \cdot (q^i - 1)}\) for \(\ell \in \set{1,2}\), hence using \cref{prop:count:lcd} we have
  \[
    \left\lvert \mathcal{X}_{n,k,q}^{h} \right\rvert =
    \begin{cases}
      C_1 q^\frac{t \cdot (n-t) - 1}{2} B_- \binomq{\frac{n}{2}-1}{\frac{t-1}{2}}_{q^2} & \text{if \(t\) odd and \(n\) even}, \\
      C_2 q^{\frac{(t+1) \cdot (n-t)}{2} - h} \binomq{\frac{n-1}{2}}{\frac{t-1}{2}}_{q^2} & \text{if \(t\) odd and \(n\) odd}, \\
      C_1 q^\frac{t \cdot (n-t+1)}{2} \binomq{\frac{n-1}{2}}{\frac{t}{2}}_{q^2} & \text{if \(t\) even and \(n\) odd}, \\
      C_2 q^\frac{t \cdot (n-t)}{2} B_+ \binomq{\frac{n}{2}}{\frac{t}{2}}_{q^2} & \text{if \(t\) even and \(n\) even},
    \end{cases}
  \]
  where \(B_- = q^\frac{n}{2} - \chi_q(-1)^\frac{n}{2}\) and \(B_+ = \frac{q^{\frac{n}{2} - h} + \chi_q(-1)^\frac{n}{2}}{q^\frac{n}{2} + \chi_q(-1)^\frac{n}{2}}\).
\end{proposition}

\subsection{Sesquilinear Forms}\label{sec:prelim:sesqui}

A map \(\Phi \colon \FF_q^n \times \FF_q^n \to \FF_q\) additive in both components is called a \emph{sesquilinear form} if there exists an automorphism \(\phi \in \Gal(\FF_q/\FF_p)\) such that \(\Phi(a \vec{x}, b \vec{y}) = a \phi(b) \Phi(\vec{x}, \vec{y})\) and further Hermitian if also \(\Phi(\vec{x}, \vec{y}) = \phi(\Phi(\vec{y}, \vec{x}))\) for any \(\vec{x}, \vec{y} \in \FF_q^n\) and \(a, b \in \FF_q\). Orthogonality will be defined analogously to the case of bilinear forms. We thus consider the following.

\begin{definition}[Galois inner products {\cite{DCC:FanZha17}}]\label{def:galois:inner}
  For \(\phi_i \in \Gal(\FF_q/\FF_p)\) define the \emph{\(i\)-th Galois inner product} as
  \[
    \inner{\vec{x}}{\vec{y}}_{\phi_i} = \inner{\vec{x}}{\phi_i(\vec{y})} = \sum_{j=1}^{n} x_j \phi_i(y_j).
  \]
\end{definition}

The basic properties of these forms are immediate.

\begin{proposition}\label{prop:galois:properties}
  Each form \(\inner{\cdot}{\cdot}_{\phi_i}\) for \(i \in \ZZ_d\) is \(\phi_i\)-sesquilinear, \(\FF_q\)-linear in the first argument, \(\FF_p\)-linear in the second argument, and admits \(\inner{\vec{x}}{\vec{y}}_{\phi_i} = \phi_i(\inner{\vec{y}}{\vec{x}}_{\phi_{-i}})\).
  In particular, it is a Hermitian form if and only if either \(i=0\) or \(d\) is even and \(i=\frac{d}{2}\).
\end{proposition}

Due to the asymmetric nature of Galois inner products, they are sometimes defined as left and right Galois inner product, depending on which component of the inner product is twisted (ours being right). However, as we will only be interested in the orthogonality relation induced by the inner products and we have that \(\inner{\phi_i(\vec{x})}{\vec{y}} = \phi_i(\inner{\vec{x}}{\phi_{-i}(\vec{y})})\) for all \(\vec{x},\vec{y} \in \FF_q^n\) and \(i \in \ZZ_d\), we without loss of generality only focus on right Galois inner products.

\subsection{The Moore Matrix}\label{sec:prelim:matrices}

We require the following special family of matrices.

\begin{definition}[Moore matrix and determinant {\cite{BAMS:Moore96}}]\label{def:moore}
  Let \(q = p^d\) and \(n,m \in \NN\). For \(a_1,\dots,a_n \in \FF_q\), the \emph{Moore matrix} is defined as
  \begin{align*}
    \mat{M}_{m}(a_1,\dots,a_n)
    &=
    \begin{bmatrix}
      a_1 & a_1^{p} & \dots & a_1^{p^{m-1}} \\
      a_2 & a_2^{p} & \dots & a_2^{p^{m-1}} \\
      \vdots & \vdots & \ddots & \vdots \\
      a_n & a_n^{p} & \dots & a_n^{p^{m-1}}
    \end{bmatrix} \iftwocolumns{\\}{}
    \iftwocolumns{&}{}=
    \begin{bmatrix}
      a_1 & \phi_1(a_1) & \dots & \phi_{m-1}(a_1) \\
      a_2 & \phi_1(a_2) & \dots & \phi_{m-1}(a_2) \\
      \vdots & \vdots & \ddots & \vdots \\
      a_n & \phi_1(a_n) & \dots & \phi_{m-1}(a_n) \\
    \end{bmatrix}.
  \end{align*}
  When \(n=m \geq 2\) and the Moore matrix is square, the \emph{Moore determinant} is equal to
  \begin{align*}
    &\det(\mat{M}_n(a_1,\dots,a_n)) \iftwocolumns{\\}{}
    \iftwocolumns{&}{}= \prod_{i=1}^{n} \prod_{c_1,\dots,c_{i-1} \in \FF_p} (c_1 a_1 + \cdots + c_{i-1} a_{i-1} + a_i).
  \end{align*}
  The Moore determinant is non-zero if and only if \(a_1,\dots,a_n \in \FF_q\) are \(\FF_p\)-linearly independent.
\end{definition}

\section{Sampling Codes from a Type, Uniformly}\label{sec:hollow}

\subsection{Additional Geometry over Finite Fields}

We consider two geometric objects over finite fields that we will use when constructing our sampling algorithm for uniformly random linear codes with a prescribed type. We first derive a simple lower bound on the probability that a linear code over \(\FF_q\) is LCD that depends only on the field \(\FF_q\), then interpret the set of self-orthogonal vectors of a linear code as a variety over \(\FF_q\) and consider when a linear code contains non-degenerate planes.

\subsubsection{The LCD Factor}

\begin{lemma}\label{lem:prob-hull}
  Let \(n,k,q \in \ZZ\) with \(0 \leq k \leq n\) and \(q\) odd prime power. The probability that a uniformly random \([n,k]\)-linear code over \(\FF_q\) has a trivial hull is bounded from below by \(1 - \frac{1}{q-1}\). In particular, this is always at least \(\frac{1}{2}\).
\end{lemma}

\begin{proof}
  The codes \(\set{\vec{0}^n}\) and \(\FF_q^n\) are always LCD, hence we assume \(0 < k < n\). Since \(\code\) is LCD if and only if \(\code^\perp\) is LCD, we may further assume \(0 < k \leq \frac{n}{2}\). We use the mass formula from~\cite[Thm.~31 \& Cor.~32]{ITIT:CMTQ19}. Denote the probability in question by \(p_{n,k,q}^0 = \frac{\mathcal{X}_{n,k,q}^0}{\mathcal{X}_{n,k,q}}\). We use the standard fact that \(\binomq{n}{k}_{q} = q^{k \cdot (n-k)} \prod_{i=1}^{k} \frac{1 - q^{-(n-k+i)}}{1 - q^{-i}}\) and the analogous formula for \(\binomq{n}{k}_{q^2}\). We will bound the denominators simply using \(1 - q^{-m} \leq 1\) for \(m \in \NN\) and the numerators by observing that for any \(s \geq 1\) we have
  \begin{align*}
    \prod_{i=1}^{s} (1 - q^{-i})
    &\geq \prod_{i=1}^{\infty} (1 - q^{-i})
    \overset{\bullet}{\geq} 1 - \sum_{i=1}^{\infty} q^{-i} \iftwocolumns{\\}{}
    \iftwocolumns{&}{}= 1 - \frac{q^{-1}}{1 - q^{-1}}
    = 1 - \frac{1}{q-1}
  \end{align*}
  by the Weierstrass product inequality. We need to consider the four cases of \cref{prop:count:lcd}.

  \paragraph*{Case \(k\) odd and \(n\) odd}
  We write \(N = \frac{n - 1}{2}\) and \(K = \frac{k - 1}{2}\). \cref{prop:count:lcd} gives \(\lvert \mathcal{X}_{n,k,q}^0 \rvert = q^\frac{(k+1) \cdot (n-k)}{2} \binomq{N}{K}_{q^2}\) and we have \(\binomq{N}{K}_{q^2} = q^{2 K \cdot (N - K)} \prod_{i=1}^{K} \frac{1 - q^{-2 \cdot (N - K + i)}}{1 - q^{-2 i}}\). Since \(\frac{(k+1) \cdot (n-k)}{2} + 2 K \cdot (N-K) = k \cdot (n-k)\), the positive power of \(q\) cancels in \(p_{n,k,q}^0\), and so do the factors \(1 - q^{-a}\) for even \(a\), so we get
  \begin{align*}
    p_{n,k,q}^0
    &= \prod_{i=1}^{K} \frac{1 - q^{-(2 i - 1)}}{1 - q^{-(n - k + 2 i - 1)}} \cdot \frac{1 - q^{-(2 K + 1)}}{1 - q^{-(n - k + 2 K + 1)}} \iftwocolumns{\\}{}
    \iftwocolumns{&}{}\overset{\bullet}{\geq} 1 - \frac{1}{q-1},
  \end{align*}
  where for \(\bullet\) we bound the numerators and denominators as discussed above.

  \paragraph*{Case \(k\) even and \(n\) odd}
  We write \(N = \frac{n-1}{2}\) and \(K = \frac{k}{2}\). \cref{prop:count:lcd} gives \(\lvert \mathcal{X}_{n,k,q}^0 \rvert = q^\frac{k \cdot (n-k+1)}{2} \binomq{N}{K}_{q^2}\) and we have \(\binomq{N}{K}_{q^2} = q^{2 K \cdot (N - K)} \prod_{i=1}^{K} \frac{1 - q^{-2 \cdot (N - K + i)}}{1 - q^{-2 i}}\). Since \(\frac{k \cdot (n-k+1)}{2} + 2 K \cdot (N-K) = k \cdot (n-k)\), the positive power of \(q\) cancels in \(p_{n,k,q}^0\), and so do the factors \(1 - q^{-a}\) for even \(a\), so we get
  \[
    p_{n,k,q}^0
    = \prod_{i=1}^{K} \frac{1 - q^{-(2 i - 1)}}{1 - q^{-(n - k + 2 i)}}
    \overset{\bullet}{\geq} 1 - \frac{1}{q-1},
  \]
  where for \(\bullet\) we bound the numerators and denominators as discussed above.

  \paragraph*{Case \(k\) even and \(n\) even}
  We write \(N = \frac{n}{2}\) and \(K = \frac{k}{2}\). \cref{prop:count:lcd} gives \(\lvert \mathcal{X}_{n,k,q}^0 \rvert = q^\frac{k \cdot (n-k)}{2} \binomq{N}{K}_{q^2}\) and we have \(\binomq{N}{K}_{q^2} = q^{2 K \cdot (N - K)} \prod_{i=1}^{K} \frac{1 - q^{-2 \cdot (N - K + i)}}{1 - q^{-2 i}}\). The positive power of \(q\) cancels in \(p_{n,k,q}^0\), and so do the factors \(1 - q^{-a}\) for even \(a\), so we get
  \[
    p_{n,k,q}^0
    = \prod_{i=1}^{K} \frac{1 - q^{-(2 i - 1)}}{1 - q^{-(n - k + 2 i - 1)}}
    \overset{\bullet}{\geq} 1 - \frac{1}{q-1},
  \]
  where for \(\bullet\) we bound the numerators and denominators as discussed above.

  \paragraph*{Case \(k\) odd and \(n\) even}
  This is the only anomalous case, as we have the additional factor \(A = q^\frac{n}{2} - \chi_q(-1)^\frac{n}{2}\). We write \(N = \frac{n}{2}\) and \(K = \frac{k-1}{2}\). \cref{prop:count:lcd} gives \(\lvert \mathcal{X}_{n,k,q}^0 \rvert = q^\frac{k \cdot (n-k) - 1}{2} \cdot \left(q^N - \chi_q(-1)^N\right) \cdot \binomq{N-1}{K}_{q^2}\) and we have \(\binomq{N-1}{K}_{q^2} = q^{2 K \cdot (N - 1 - K)} \prod_{i=1}^{K} \frac{1 - q^{-2 \cdot (N - 1 - K + i)}}{1 - q^{-2 i}}\). We use \(q^N - \chi_q(-1)^N = q^N \cdot (1 - \chi_q(-1)^N \cdot q^{-N})\) to get
  \[
    p_{n,k,q}^0
    = \frac{1 - \chi_q(-1)^N q^{-N}}{1 - q^{-n}} \prod_{i=1}^{K} \frac{1 - q^{-(2 i - 1)}}{1 - q^{-(n - k + 2 i)}}.
  \]
  When \(\chi_q(-1)^N = -1\), we have \(\frac{1 + q^{-N}}{1 - q^{-n}} = \frac{1}{1 - q^{-N}} \geq 1\) and bound as in the other cases. Otherwise \(\chi_q(-1)^N = 1\) and we have for \(N = 1\) that \(p_{2,1,q}^0 = \frac{1 - q^{-1}}{1 - q^{-2}} = \frac{q}{q+1} \geq 1 - \frac{1}{q-1}\). When \(N \geq 2\), we have \(\frac{1 - q^{-N}}{1 - q^{-n}} = \frac{1}{1 + q^{-N}} \geq \frac{1}{1 + q^{-2}}\). The product of odd negative powers in the numerator can be better bounded by \(1 - \frac{q}{q^2 - 1}\), and we finally get
  \begin{align*}
    \frac{1}{1 + q^{-2}} \cdot \left( 1 - \frac{q}{q^2 - 1} \right)
    &= \frac{q^2}{q^2 + 1} \cdot \left( 1 - \frac{q}{q^2 - 1} \right) \iftwocolumns{\\}{}
    \iftwocolumns{&}{}\geq 1 - \frac{1}{q-1}. \qedhere
  \end{align*}
\end{proof}

Since the above probability bound will arise many times during our algorithmic treatment and depends only on the underlying field, we introduce notation for it.

\begin{definition}[LCD factor]\label{def:lcd-factor}
  For a finite field \(\FF_q\) denote its \emph{LCD factor} by \(\gamma_q = 1 - \frac{1}{q-1}\).
\end{definition}

By \cref{lem:prob-hull}, the LCD factor of a field \(\FF_q\) is thus a rough but useful lower bound on the probability that a uniformly random full-rank matrix over \(\FF_q\) of appropriate dimensions is \(0\)-hollow. In other words, we have the following.

\begin{corollary}\label{cor:prob-hull}
  Let \(n,k,q \in \ZZ\) with \(0 \leq k \leq n\) and \(q\) prime power. Then \(\probsub{\mat{A} \sample \mathsf{U}_{n,k,q}}{\det(\mat{A}\mat{A}^\transpose) \neq 0} \geq \gamma_q\).
\end{corollary}

\subsubsection{The Variety of Self-Orthogonal Vectors}

While the hull is the largest self-orthogonal subcode of a linear code, the set of self-orthogonal vectors without any linear structure is much greater in size and has the structure of a variety in the coefficient vector space \(\FF_q^k\).

\begin{proposition}\label{prop:so-variety}
  Let \(n,k,q \in \ZZ\) with \(0 < k \leq n\) and \(q\) prime power. Given a generator matrix \(\mat{A}\) for an \([n,k]\)-linear code \(\code\) over \(\FF_q\), a vector \(\vec{v} = \mat{A}^\transpose \vec{y}\) is self-orthogonal in \(\code\) if and only if \(\vec{y}\) is a root of the polynomial \(f_\mat{A}(\vec{x}) = \vec{x}^\transpose \mat{A} \mat{A}^\transpose \vec{x} \in \FF_q[X_1,\dots,X_k]\) where \(\vec{x} = [X_1,\dots,X_k]^\transpose\) is a vector of \(k\) variables over \(\FF_q\).
\end{proposition}

The variety \(\variety_{f_\mat{A}}\) of course depends on the generator matrix \(\mat{A}\) for \(\code\), however, we care less about the coefficient vectors \(\vec{x} \in \variety_{f_\mat{A}}\) but rather the self-orthogonal vectors \(\mat{A}^\transpose \vec{x} \in \code\), the set which does not depend on \(\mat{A}\).

\begin{definition}\label{def:code:variety}
  Denote by \(\variety_\code\) the set of self-orthogonal vectors in a linear code \(\code\). For any generator matrix \(\mat{A}\) for \(\code\) we have \(\variety_\code = \condset{\mat{A}^\transpose \vec{x}}{\vec{x} \in \variety_{f_\mat{A}}}\).
\end{definition}

Regardless of the hull dimension of \(\code\), the set \(\variety_\code\) contains exponentially many vectors in the dimension of \(\code\).

\begin{lemma}\label{lem:code:variety:nsols}
  For every \([n,k]\)-linear code \(\code\) over \(\FF_q\) we have \(\left\lvert \variety_\code \right\rvert \geq q^{k-2}\).
\end{lemma}

\begin{proof}
  Let \(\mat{A} \in \FF_q^{k \times n}\) generate \(\code\) and assume notation of \cref{prop:so-variety}. The \(k\)-variate polynomial \(f_\mat{A}\) of degree \(2\) over \(\FF_q\) is homogeneous, hence it has a root \(\vec{0}^k \in \variety_{f_\mat{A}}\). If \(k > 2\), by Warning's Second Theorem~\cite[Satz~3]{HAM:Warning35} then \(\left\lvert \variety_{f_\mat{A}} \right\rvert \geq q^{k-2}\) and hence by \cref{def:code:variety} also \(\left\lvert \variety_\code \right\rvert \geq q^{k-2}\). For \(k \leq 2\) it follows from \(\vec{0}^k \in \variety_\code\) that \(\left\lvert \variety_\code \right\rvert \geq 1 \geq q^{k-2}\).
\end{proof}

\subsection{Regular Planes}

While a linear code spanned by \(\mat{A}\) possesses many self-orthogonal vectors, which can be seen as self-orthogonal lines, there are also non-degenerate subspaces of higher dimensions if the hull is not too big, such as non-degenerate planes.

\begin{definition}[Regular plane]\label{def:hollow:plane}
  Let \(n, q \in \ZZ\) with \(n \geq 2\) and \(q\) prime power. A \emph{regular plane} is an LCD code over \(\FF_q\) of dimension \(2\) and length \(n\).
  A regular plane generated by \(\mat{T}\) is \emph{hyperbolic} if \(V_\mat{T}\) is isometric to a hyperbolic plane, equivalently \(\chi_\mat{T} = \chi_q(-1)\).
\end{definition}

\begin{lemma}\label{lem:hollow:plane}
  A linear code over \(\FF_q\) for odd prime power \(q\) contains a regular plane if and only if it has hull co-dimension at least \(2\), and further contains a hyperbolic plane if and only if it has hull co-dimension at least \(3\) or hull co-dimension \(2\) and square class \(\chi_q(-1)\).
\end{lemma}

\begin{proof}
  Direct corollary of~\cite[Sec.~III.6]{Artin1957}. Any linear code with hull co-dimension at least \(3\) has an LCD component \(\mat{T}\) such that \(V_\mat{T}\) is regular. Then \(V_\mat{T}\) is isometric to an orthogonal sum of subspaces with at least one hyperbolic plane, which corresponds to a basis change of \(\mat{T}\) that has the first two rows spanning a hyperbolic plane. When hull co-dimension is \(2\), we have that \(V_\mat{T}\) is either isometric to a hyperbolic plane, when \(\chi_\mat{T} = \chi_q(-1)\), or to an anisotropic space, which corresponds to a non-hyperbolic regular plane. When hull co-dimension is \(1\) there are no regular planes, only one regular line for each hull vector.
\end{proof}

We consider hyperbolic planes in particular because they always contain a non-zero self-orthogonal vector, which will be useful when trying to sample such vectors.

\begin{proposition}\label{prop:plane:hyper}
  A regular plane over \(\FF_q\) for odd \(q\) contains a non-zero self-orthogonal vector if and only if it is hyperbolic.
\end{proposition}

\begin{proof}
  Let \(\mat{T}\) span the plane. By~\cite[Sec.~III.6]{Artin1957}, there are exactly two possibilities for \(V_\mat{T}\) up to isometry, these are either the hyperbolic plane when \(\chi_\mat{T} = \chi_q(-1)\), or the anisotropic representative. Since anisotropic spaces contain no non-zero isotropic vectors and isotropic vectors in \(V_\mat{T}\) are in bijective correspondence with self-orthogonal vectors in \(\code(\mat{T})\) through \(\vec{x} \mapsto \mat{T}^\transpose \vec{x}\), the conclusion follows.
\end{proof}

\subsection{How to Sample a Hull}\label{sec:hollow:sample}

We provide an explicit algorithm for finding regular planes in linear codes, use it to construct an algebraic sampling algorithm that samples uniformly random self-orthogonal vectors of a code, and then iteratively build a sampling algorithm for uniformly random linear codes of a given hull dimension and square class.

\subsubsection{Sampling LCD Codes}

We start with LCD codes, which are easy to sample using rejection sampling, since most linear codes are LCD by \cref{cor:prob-hull}. Without loss of generality, we only deal with codes of rate at most \(\frac{1}{2}\), since a code is LCD if and only if its dual is LCD. This makes the analysis significantly easier.

\begin{definition}[Algorithm \(\SampleLCD\)]\label{def:algo:samplelcd}
  Let \(n,k,q,\chi \in \ZZ\) with \(0 \leq k \leq \frac{n}{2}\), \(\chi \in \set{\pm 1}\), and \(q\) odd prime power. Consider the algorithm in \cref{fig:algo:samplelcd} that on input \(n, k, q, \chi\) outputs a code generator matrix \(\mat{A}\) with \(\code(\mat{A}) \in \mathcal{X}_{n,k,q}^{\chi,0}\) or aborts with \(\bot\). If \(k = 0\), the algorithm always outputs the zero code \(\set{\vec{0}^n}\).
\end{definition}

\begin{figure}[t]
  \centering
  \begin{pchstack}[boxed]
    \procedure[mode=text]{$\SampleLCD(n,k,q,\chi)$}{
      \pcif \(k = 0\)\pcthen \pcreturn \([\cdot] \in \FF_q^{0 \times n}\) \\
      \pcfor \(\_ = 1,\dots,k\) \pcdo \\
      \t \(\mat{A} \sample \FF_q^{k \times n}\) \\
      \t \pcif \(\rank(\mat{A}) \neq k\)\pcthen \pccontinue \\
      \t \pcif \(\rank(\mat{A}\mat{A}^\transpose) \neq k\)\pcthen \pccontinue \\
      \t \pcif \(\chi_\mat{A} \neq \chi\)\pcthen \pccontinue \\
      \t \pcreturn \(\mat{A}\) \\
      \pcreturn \(\bot\)
    }
  \end{pchstack}
  \caption{The algorithm \(\SampleLCD\) that, given \(n,k,q,\chi \in \ZZ\) outputs a matrix \(\mat{A} \in \FF_q^{k \times n}\) with \(\code(\mat{A}) \in \mathcal{X}_{n,k,q}^{\chi,h}\).}
  \label{fig:algo:samplelcd}
\end{figure}

\begin{lemma}\label{lem:algo:samplelcd}
  Assume the notation of \cref{def:algo:samplelcd}. The algorithm \(\SampleLCD\) runs in time \(\bigO{k \cdot (n^\omega + \log^2(q))}\), succeeds with probability overwhelming in \(k\), and its output distribution conditioned on not aborting is \(\dist_{n,k,q}^{\chi,0}\).
\end{lemma}

\begin{proof}
  The algorithm is in \cref{fig:algo:samplelcd} and is clearly correct if it succeeds with a matrix \(\mat{A}\), as we are simply rejecting generator matrices with \(\code(\mat{A}) \notin \mathcal{X}_{n,k,q}^{\chi,h}\). The running time is immediate, since residue testing takes \(\log^2(q)\) by~\cite{Cohen1993}, so we focus on the success probability. The probability that a uniformly random matrix \(\mat{A} \sample \FF_q^{k \times n}\) is full-rank is \(p_\mathsf{FR} = \prod_{i=0}^{k-1} (1 - q^{i-n}) \geq 1 - \frac{k}{q^k}\), bound by the Weierstrass product inequality. Next, the probability that a uniformly random full-rank matrix generates an LCD code is at least \(\gamma_q\) by \cref{cor:prob-hull}. The probability that a uniformly random \(0\)-hollow matrix generates a code with square class \(\chi\) is by~\cite[Thm.~1]{localinv} at least \(\frac{1}{2} - q^{-\frac{k}{2}}\). The algorithm fails if in every one of the \(k\) repetitions at least one of these events does not happen. Using that \(\gamma_q \geq \frac{1}{2}\) by \cref{lem:prob-hull} for \(\bullet\), this probability is at most
  \begin{align*}
    \iftwocolumns{&}{}\left(1 - \left(1 - k q^{-k}\right) \cdot \gamma_q \cdot \left(\frac{1}{2} - q^{-\frac{k}{2}}\right)\right)^k \iftwocolumns{\\}{}
    &\overset{\bullet}{\leq} \left( 1 - \frac{1}{4} + \frac{k}{4} q^{-k} + \frac{1}{2} q^{-\frac{k}{2}} - \frac{k}{2} q^{-\frac{3k}{2}} \right)^k \\
    &\leq \left( \frac{3}{4} + \negl[k] \right)^k
    \leq \negl[k]. \qedhere
  \end{align*}
\end{proof}

\begin{remark}\label{rem:algo:samplelcd}
  An algorithm that samples uniformly random LCD codes without considering the square class can be obtained by removing the line that checks the square class. Alternatively we can instead sample \(\chi \sample \set{\pm 1}\) instead of taking it as input and obtain an algorithm whose output distribution is negligibly close in \(k\) to the uniform distribution on LCD codes.
\end{remark}

\subsubsection{Finding LCD Planes}

Our sampler will start with a uniformly random LCD code of a desired square class and endow it with a hull of the specified dimension. We sample the hull by iteratively sampling the next hull vector as a uniformly random self-orthogonal vector from the dual of the current linear code. Adding hull vectors of course preserves the square class, so the resulting code will indeed be of the correct type.
Sampling self-orthogonal vectors will prove to be easy and efficient, all we need to do is find an appropriate parametrisation of the variety \(\variety_\code\) of a linear code \(\code\). The parametrisation will come in the form of a family of smooth conics, where the non-degeneracy condition will be fulfilled by first finding a regular plane on which to base our sampler.

\begin{definition}[Algorithm \(\FindPlane\)]\label{def:algo:findplane}
  Let \(n,k,h,q \in \ZZ\) with \(0 \leq h \leq k \leq n - h\), \(q\) odd prime power, and additionally \(h \leq k - 2\). Consider the algorithm \(\FindPlane\) in \cref{fig:algo:findplane} that on input an \(h\)-hollow matrix \(\mat{A} \in \FF_q^{k \times n}\) and indices \(i^* \neq j^*\) in \(\llbracket 1, k \rrbracket\), outputs a matrix \(\mat{B} \in \FF_q^{k \times n}\) such that \(\code(\mat{A}) = \code(\mat{B})\), and \(\mat{B}[i^*]\) and \(\mat{B}[j^*]\) span a regular plane.
\end{definition}

\begin{figure}[t]
  \centering
  \begin{pchstack}[boxed]
    \procedure[mode=text]{$\FindPlane(\mat{A}, i^*, j^*)$}{
      \(\mat{G} \coloneqq \mat{A} \mat{A}^\transpose\);\;
      \pcassert \(\rank(\mat{G}) \geq 2\);\;
      \pcassert \(i^* \neq j^*\) \\
      \pcif \(\mat{G}_{i,i} = 0\) for \(i=1,\dots,k\)\pcthen \\
      \t pick any \((i, j)\) such that \(\mat{G}_{i,j} \neq 0\) \\
      \t \pcreturn \(\mat{A}\) permuted s.t. \(i \mapsto i^*\) and \(j \mapsto j^*\) \\
      pick any \(i\) such that \(\mat{G}_{i,i} \neq 0\) \\
      \(\mat{A} \coloneqq \mat{A}\) with transposed rows \(1\) and \(i\) \\
      \pcfor \(i=2,\dots,k\)\pcdo \\
      \t \(\mat{A}[i] \coloneqq \mat{A}[i] - \frac{\inner{\mat{A}[1]}{\mat{A}[i]}}{\inner{\mat{A}[1]}{\mat{A}[1]}} \mat{A}[1]\) \\
      \(\mat{A}' \coloneqq \mat{A}\llbracket 2, k \rrbracket\);\;
      \(\mat{G}' \coloneqq \mat{A}' \mat{A'}^\transpose\) \\
      \pcif \(\mat{G}_{i,i}' = 0\) for \(i=1,\dots,k-1\)\pcthen \\
      \t pick any \((i,j)\) such that \(\mat{G}_{i,j}' \neq 0\) \\
      \t \pcreturn \(\mat{A}\) permuted s.t. \(i+1 \mapsto i^*\) and \(j+1 \mapsto j^*\) \\
      pick any \(i\) such that \(\mat{G}_{i,i}' \neq 0\) \\
      \pcreturn \(\mat{A}\) permuted s.t. \(1 \mapsto i^*\) and \(i+1 \mapsto j^*\)
    }
  \end{pchstack}
  \caption{The algorithm \(\FindPlane\) that, given any \(h\)-hollow \(\mat{A} \in \FF_q^{k \times n}\) with \(k-h \geq 2\) finds \(\mat{B}\) generating the same code such that \(\mat{B}[i^*]\) and \(\mat{B}[j^*]\) span a regular plane.}
  \label{fig:algo:findplane}
\end{figure}

\begin{lemma}\label{lem:algo:findplane}
  The algorithm \(\FindPlane\) runs in time \(\bigO{n^\omega}\) and always succeeds.
\end{lemma}

\begin{proof}
  Without loss of generality, suppose in this proof \(i^* = 1\) and \(j^* = 2\). Suppose for all rows \(\vec{v}^\transpose\) of \(\mat{A}\) we have \(\inner{\vec{v}}{\vec{v}} = 0\). Then there must exist two rows \(\vec{v}^\transpose, \vec{w}^\transpose\) of \(\mat{A}\) such that \(\inner{\vec{v}}{\vec{w}} \neq 0\), otherwise \(\mat{A}\) spans a self-orthogonal code. Thus \(\vec{v}\) and \(\vec{w}\) span a regular plane, so our algorithm simply permutes the rows of \(\mat{A}\) into \(\mat{B}\) such that \(\mat{B}\) starts with \(\vec{v}\) and \(\vec{w}\), and clearly \(\code(\mat{A}) = \code(\mat{B})\).

  Otherwise let \(\vec{v}^\transpose\) be the first row of \(\mat{A}\) with \(\inner{\vec{v}}{\vec{v}} \neq 0\). Compute a basis for the orthogonal complement of \(\vec{v}\) in \(\code(\mat{A})\), that is, a generator matrix \(\mat{A}' \in \FF_q^{(k-1) \times n}\) for \(\code \cap \code(\vec{v}^\transpose)^\transpose\) using the usual projections. Again, if for all rows \(\vec{w}^\transpose\) of \(\mat{A}'\) we have \(\inner{\vec{w}}{\vec{w}} = 0\) then find any two rows of \(\mat{A}'\) that are not orthogonal, since it spans a code of hull co-dimension at least \(1\), and return \([\mat{A}' \Mid \vec{v}^\transpose]\) permuted such that the two non-orthogonal rows of \(\mat{A}'\) are first. Otherwise let \(\vec{w}^\transpose\) be the first row that is not self-orthogonal in \(\mat{A}'\). Then \(\vec{v}\) and \(\vec{w}\) span a regular plane, hence return \([\vec{v}^\transpose \Mid \mat{A}']\) permuted such that \(\mat{A}'\) starts with \(\vec{w}^\transpose\).

  The first condition can be checked with linear algebra by computing the Gram matrix \(\mat{G} = \mat{A}\mat{A}^\transpose\) and checking whether it has all \(0\)s on the diagonal, and if so, finding \((i,j)\) such that \(\mat{G}_{i,j} \neq 0\) and permuting accordingly. Since we already have all the inner products from computing the Gram matrix, orthogonalisation costs \(k-1\) vector additions in \(\FF_q^n\).
\end{proof}

While any regular plane would allow us to sample uniformly random self-orthogonal vectors from a fixed linear code, our plan is to iteratively increase the dimension of our code. Thus our sampling algorithm needs to be independent of the linear code. We ensure this by always picking a hyperbolic plane. A hyperbolic plane always exists by \cref{lem:hollow:plane} as long as there are self-orthogonal vectors outside of the hull, otherwise we simply sample from the hull directly.

\begin{definition}[Algorithm \(\FindPlaneH\)]\label{def:algo:findplane:hyper}
  Let \(n,k,h,q \in \ZZ\) with \(0 \leq h \leq k \leq n - h\), \(q\) odd prime power, and additionally \(h \leq k - 2\). Consider the algorithm \(\FindPlaneH\) in \cref{fig:algo:findplane:hyper} that on input an \(h\)-hollow matrix \(\mat{A} \in \FF_q^{k \times n}\) and indices \(i^* \neq j^*\) in \(\llbracket 1, k \rrbracket\), outputs a matrix \(\mat{B} \in \FF_q^{k \times n}\) such that \(\code(\mat{A}) = \code(\mat{B})\), and \(\mat{B}[i^*]\) and \(\mat{B}[j^*]\) span a hyperbolic plane, or aborts with \(\bot\).
\end{definition}

\begin{figure}[t]
  \centering
  \begin{pchstack}[boxed]
    \procedure[mode=text]{$\FindPlaneH(\mat{A}, i^*, j^*)$}{
      \pcassert \(i^* \neq j^*\);\;
      \(\mat{T}, \mat{H} \coloneqq \Dehull(\mat{A})\);\;
      \(t \coloneqq \algo{rank}(\mat{T})\) \\
      \pcassert \(t \geq 2\);\;
      \pcif \(t = 2\) \pcand \(\chi_\mat{T} \neq \chi_q(-1)\)\pcthen \pcreturn \(\bot\) \\
      \(\mat{T}' \coloneqq \FindPlane(\mat{T}, t-1, t)\) \\
      \pcif \(\chi_{\mat{T}'\llbracket t-1,t\rrbracket} = \chi_q(-1)\)\pcthen \\
      \t \pcreturn \([\mat{T}' \Mid \mat{H}]\) permuted s.t. \(t-1 \mapsto i^*\) and \(t \mapsto j^*\) \\
      \(f \coloneqq \vec{x}^\transpose \mat{T}' {\mat{T}'}^\transpose \vec{x}\);\;
      \pcfor \(i \coloneqq 1,\dots,k-2\) \pcdo\ \(c_i \sample \FF_q^\star\) \\
      \(F(x, y) \coloneqq f(c_1,\dots, c_{k-2}, x, y)\) \\
      \((c_{k-1}, c_k) \sample \variety_F\);\;
      \(\vec{c} \coloneqq [c_i]_{i=1}^k\);\;
      \(\vec{v} \coloneqq {\mat{T}'}^\transpose \vec{c}\) \\
      pick \(\vec{y}\) s.t. \(\vec{v}^\transpose {\mat{T}'}^\transpose \vec{y} = 1\);\;
      \(\vec{z} \coloneqq \vec{y} - \frac{\vec{y}^\transpose \mat{T}' {\mat{T}'}^\transpose \vec{y}}{2} \vec{c}\) \\
      \(\vec{w} \coloneqq {\mat{T}'}^\transpose \vec{z}\);\;
      basis change \(\mat{T}'\) s.t. it starts with \([\vec{v} \mid \vec{w}]^\transpose\) \\
      \pcreturn \([\mat{T}' \Mid \mat{H}]\) permuted s.t. \(1 \mapsto i^*\) and \(2 \mapsto j^*\)
    }
  \end{pchstack}
  \caption{The algorithm \(\FindPlaneH\) attempts to change the output of \(\FindPlane\) in \protect\cref{fig:algo:findplane} to a hyperbolic plane.}
  \label{fig:algo:findplane:hyper}
\end{figure}

\begin{lemma}\label{lem:algo:findplane:hyper}
  The algorithm \(\FindPlaneH\) runs in time \(\bigO{n^\omega + \log^2(q)}\) and always succeeds, except when the input matrix \(\mat{A}\) has hull co-dimension \(2\) and \(\chi_\mat{A} \neq \chi_q(-1)\), in which case it always aborts.
\end{lemma}

\begin{proof}
  The algorithm first finds any LCD component \(\mat{T}\) of the input \(\mat{A}\), runs \(\FindPlane\) and checks whether the obtained regular plane is already hyperbolic. Otherwise, we use the obtained regular plane to find any self-orthogonal vector in the span of \(\mat{T}\) that is not already in the hull, for which it suffices to sample non-zero coefficients \(\vec{c}[1],\dots,\vec{c}[t-2]\) and pick any solution of the obtained affine conic, which exists by \cref{prop:conic}. Once we have a self-orthogonal vector \(\vec{v} = \mat{T}^\transpose \vec{c}\) outside of the hull, constructing a hyperbolic plane is simple: pick any coefficient vector \(\vec{y}\) such that \(\vec{v}^\transpose \mat{T}^\transpose \vec{y} = 1\) using linear algebra, orthogonally project it using \(\vec{z} = \vec{y} - \frac{\vec{y}^\transpose \mat{T}\mat{T}^\transpose \vec{y}}{2} \vec{c}\) and compute \(\vec{w} = \mat{T}^\transpose \vec{z}\). It is immediate that both \(\vec{v}\) and \(\vec{w}\) are self-orthogonal, we have \(\inner{\vec{v}}{\vec{w}} = 1\), and so \([\vec{v} \mid \vec{w}]^\transpose\) generates a hyperbolic plane. Then simply compute a basis change of \(\mat{T}\) such that it starts with \(\vec{v}\) and \(\vec{w}\), append the missing hull rows and permute the basis appropriately.
\end{proof}

\subsubsection{Sampling Self-Orthogonal Vectors}

Now that we know how to find hyperbolic planes, we will use them to construct a family of smooth conics whose points correspond bijectively to elements of \(\variety_\code\) for a given linear code \(\code\). This will directly mirror how we construct a self-orthogonal vector in \cref{def:algo:findplane:hyper}, except that we use \cref{prop:conic} to weight the output of our algorithm appropriately to sample the uniform distribution on \(\variety_\code\) with rejection.

\begin{definition}[Algorithm \(\SampleSelfOrthogonal\)]\label{def:algo:sso}
  Let \(n, k, h, q \in \ZZ\) with \(0 \leq h < k \leq n - h\), \(h \leq k - 2\), and \(q\) odd prime power. Consider the algorithm in \cref{fig:algo:sso} that on input an \(h\)-hollow \(\mat{A} \in \FF_q^{k \times n}\) outputs a self-orthogonal vector from \(\variety_{\code(\mat{A})}\) or aborts with \(\bot\).
\end{definition}

\begin{figure}[t]
  \centering
  \begin{pchstack}[boxed]
    \procedure[mode=text]{$\SampleSelfOrthogonal(\mat{A})$}{
      \(\mat{A}' \coloneqq \FindPlaneH(\mat{A},k-1,k)\) \\
      \pcif \(\mat{A}' = \bot\)\pcthen \pcreturn \(\vec{v} \sample \hull(\mat{A})\) \pcelse \(\mat{A} \coloneqq \mat{A}'\) \\
      \(f \coloneqq \vec{x}^\transpose \mat{A} \mat{A}^\transpose \vec{x}\);\;
      \(M \coloneqq 2 q - 1\) \\
      \pcfor \(\_ \coloneqq 1,\dots,n\) \pcdo \\
      \t \([c_1,\dots,c_{k-2}]^\transpose \sample \FF_q^{k-2}\);\;
      \(F(x, y) \coloneqq f(c_1,\dots, c_{k-2}, x, y)\) \\
      \t \(\var{rs} \coloneqq \variety_F\);\;
      \(u \sample \Unif{[0,1]}\);\;
      \pcif \(u > \frac{\lvert \var{rs} \rvert}{M}\)\pcthen \pccontinue \\
      \t \((c_{k-1}, c_k) \sample \Unif{\var{rs}}\);\;
      \(\vec{x} \coloneqq [c_i]_{i=1}^k\);\;
      \pcreturn \(\vec{y} \coloneqq \mat{A}^\transpose \vec{x}\) \\
      \pcreturn \(\bot\)
    }
  \end{pchstack}
  \caption{The algorithm \(\SampleSelfOrthogonal\) (\(\mathsf{S}\)ample\(\mathsf{S}\)elf\(\mathsf{O}\)rthogonal) that, given a code generator matrix \(\mat{A} \in \FF_q^{k \times n}\), outputs a self-orthogonal vector in \(\code(\mat{A})\). The algorithm \(\FindPlaneH\) is in \protect\cref{fig:algo:findplane:hyper}.}
  \label{fig:algo:sso}
\end{figure}

\begin{lemma}\label{lem:algo:sso}
  Assume the notation of \cref{def:algo:sso} with both \(n,q = \poly[k]\). The algorithm \(\SampleSelfOrthogonal\) runs in time \(\poly[n,k,q]\), aborts with probability negligible in \(n\), and for any full-rank input \(\mat{A} \in \FF_q^{k \times n}\) its output distribution conditioned on not aborting is \(\Unif{\variety_{\code(\mat{A})}}\).
\end{lemma}

\begin{proof}
  The high-level proof idea is the following. The algorithm will first transform its input matrix such that its last two rows span a hyperbolic plane using \cref{lem:algo:findplane:hyper}. Let \(x, y\) be variables over \(\FF_q\). Every vector \(\vec{v} \in \FF_q^{k-2}\) defines a smooth conic \(f_\vec{v}(x, y) = \vec{x}^\transpose \mat{A} \mat{A}^\transpose \vec{x}\), where \(\vec{x} = [\vec{v} \Mid x \Mid y]\), whose points are in bijection with elements of \(\variety_{\code(\mat{A})}\) by \cref{prop:so-variety}. All conics in the family share the same discriminant independent of \(\mat{A}\), since we always pick a hyperbolic plane. Then by \cref{prop:conic} each of these conics has a non-empty zero set, however, they may be of two different possible sizes (for a fixed discriminant), thus we weight these points differently depending on which conic they lie on, allowing us to sample uniformly from \(\variety_{\code(\mat{A})}\) by simply sampling \(\vec{v} \sample \FF_q^{k-2}\) and sampling from the set of points on \(f_\vec{v}\) with potential rejection to correct the distribution. We discuss efficiency and correctness, failure probability, and uniformity of the output.

  \paragraph*{Efficiency and correctness}
  Assume the algorithm \(\SampleSelfOrthogonal\) returns a vector \(\vec{y}\), then \(\vec{y} \in \variety_\code\) by \cref{prop:so-variety}. Denote by \(\var{halt} = n\) the maximal number of loop repetitions. The algorithm first uses linear algebra (\cref{lem:algo:findplane}) to find a hyperbolic plane in the code generated by \(\mat{A}\) and compute the discriminant \(D\), then samples at most \(\var{halt} \cdot (k-2)\) elements from \(\FF_q\), solves \(\var{halt}\) conic equations over \(\FF_q\), samples once from a set with at most \(2 q - 1\) elements, and computes one multiplication of \(\mat{A}\) with the vector \([c_i]_{i=1}^k\). Computing the solutions to the conic equation na{\"\i}vely requires the computation of the Jacobi symbol, a square root~\cite[Sec.~1.5.1]{Cohen1993} and an inverse~\cite[Sec.~1.6]{Cohen1993}, each \(q\) times, together costing \(\bigO{q \log^4(q)}\). The worst-case cost of running the algorithm is therefore \(\bigO{n^\omega + \var{halt} k q \log^4(q) + k^2}\), and we may fold the term \(k^2\) into \(n^\omega\).

  \paragraph*{Failure probability}
  Since \(k-h \geq 2\), there either exists a hyperbolic plane in \(\code(\mat{A})\) by \cref{lem:hollow:plane} and we may find it efficiently using \cref{lem:algo:findplane:hyper}, or all self-orthogonal vectors lie in \(\hull(\code(\mat{A}))\). In the latter case, we return a uniformly random element of \(\hull(\code(\mat{A}))\). In the former case, assume without loss of generality that \(\inner{\mat{A}[k-1]}{\mat{A}[k-1]} \inner{\mat{A}[k]}{\mat{A}[k]} \neq \inner{\mat{A}[k-1]}{\mat{A}[k]}^2\). For every choice of \(c_1,\dots,c_{k-2}\), the function \(F(x,y)\) defines a smooth affine conic. Indeed, the coefficient \(a\) of \(x^2\) is \(\inner{\mat{A}[k-1]}{\mat{A}[k-1]}\), the coefficient \(b\) of \(x y\) is \(2 \inner{\mat{A}[k-1]}{\mat{A}[k]}\), the coefficient \(c\) of \(y^2\) is \(\inner{\mat{A}[k]}{\mat{A}[k]}\), hence the discriminant of \(F(x,y)\) is
  \begin{align*}
    D_\mat{A}
    = b^2 - 4 a c
    &= 4 \inner{\mat{A}[k-1]}{\mat{A}[k]}^2 \iftwocolumns{\\}{}
    \iftwocolumns{&}{}- 4 \inner{\mat{A}[k-1]}{\mat{A}[k-1]} \inner{\mat{A}[k]}{\mat{A}[k]},
  \end{align*}
  which is non-zero in \(\FF_q\) since \(q\) is odd. Recall that since we chose a hyperbolic plane, we have from \cref{prop:conic} that \(D_\mat{A} \in \chi_q^{-1}(1)\) and so the number of conic solutions is \(S \in \set{q-1,2 q - 1}\). The algorithm can only return \(\bot\) if for each of the \(\var{halt}\) repetitions the conic \(F(x,y)\) has \(q-1\) solutions and the algorithm sampled \(u \sample \Unif{[0,1]}\) such that \(u > \frac{q - 1}{2 q - 1}\) each time. Hence the failure probability is at most
  \[
    \left( \frac{q}{2 q - 1} \right)^\var{halt} = 2^{\log\left( \frac{q}{2 q - 1} \right) \cdot n} < 2^{-\frac{n}{2}}.
  \]

  \paragraph*{Uniform output distribution}
  To show this algorithm samples from \(\Unif{\variety_\code}\), assume again that it did not halt with \(\bot\). We must show that every self-orthogonal vector \(\vec{y} \in \variety_{\code}\) is sampled with the same probability. Since a smooth affine conic always has at least \(1\) solution, we have by \cref{def:code:variety} that \(\SampleSelfOrthogonal\) samples each vector in \(\variety_\code\) with non-zero probability. Let \(M = 2 q - 1\) be the maximum number of solutions to the obtained conic, which is \(2 q - 1\) since \(D_\mat{A} \in \chi_q^{-1}(1)\). Regardless of the number of solutions \(S\) to the obtained conic, the probability of each vector being sampled is
  \begin{align*}
    \condprob{\vec{y} \from \SampleSelfOrthogonal(\mat{A})}{\vec{y} \neq \bot}
    &= \frac{1}{q^{k-2}} \frac{1}{S} \probsub{u \sample [0,1]}{u \leq \frac{S}{M}} \iftwocolumns{\\}{}
    \iftwocolumns{&}{}= \frac{1}{q^{k-2}} \frac{1}{M}.
  \end{align*}
  Since \(M\) is fixed at the beginning of execution by the parametrisation of \(\variety_\code\) as a family of conics, the probability of each self-orthogonal vector being sampled must be the same, that is, \(\SampleSelfOrthogonal\) conditioned on not aborting samples from \(\Unif{\variety_\code}\).
\end{proof}

\begin{remark}
  For each of the cases in \cref{prop:conic} it is possible to parametrise the solution set of the conic instead of na{\"\i}vely iterating through it. The complexity of \(\SampleSelfOrthogonal\) is then polynomial in \(\log(q)\) instead of in \(q\) and we can relax our previous requirement that \(q = \poly[k]\). The same holds for \(\SampleCode\) in \cref{def:algo:samplecode}. We do not address this further in this thesis.
\end{remark}

\subsubsection{Sampling Hollow Codes}

We finally assemble the above algorithms to sample linear code generator matrices from \(\mathcal{X}_{n,k,q}^{\chi,h}\). We first sample an LCD code of the appropriate dimension and square class, then use the above algorithm iteratively on the dual code to construct a hull of the desired dimension. Uniformity is due to the following statement.

\begin{lemma}\label{lem:hollow:unif}
  Let \(n,k,h,q,\chi \in \ZZ\) with \(0 \leq h \leq k < n - h\), \(\chi \in \set{\pm 1}\) and \(q\) odd prime power. Suppose \(\mat{A} \sample \dist_{n,k,q}^{\chi,h}\), \(\vec{y} \sample \variety_{\code(\mat{A})^\perp} \setminus \hull(\code(\mat{A}))\) and \(\mat{U} \sample \GL_{k+1}(\FF_q)\), then \(\mat{A}' = \mat{U} \cdot [\mat{A} \Mid \vec{y}^\transpose]\) is distributed according to \(\dist_{n,k+1,q}^{\chi,h+1}\).
\end{lemma}

\begin{proof}
  The proof is essentially one recursive step in the proof of~\cite[Thm.~5]{ITIT:LiShiLin25}. By~\cite[Thm.~6.26]{LidlNiederreiter1997} the number of self-orthogonal vectors in \(\code(\mat{A})^\perp\) does not depend on \(\mat{A}\) but only on the parameters \(n,k,h,q,\chi\). Indeed, as derived in~\cite[Prop.~6~\&~7]{ITIT:LiShiLin25}, an LCD code of square class \(\pm 1\) contains \(M_\pm(n,k,q) = q^{k-1} - 1\) non-zero vectors if \(k\) is odd and \(M_\pm(n,k,q) = \left(q^\frac{k}{2} \mp \chi_q(-1)^\frac{k}{2}\right) \cdot \left(q^{\frac{k}{2}-1} \pm \chi_q(-1)^\frac{k}{2}\right)\) if \(k\) is even (and \(M_\pm(n,k,q) = 0\) if \(k=0\)). Thus an \([n,k]\)-linear \(h\)-hollow code \(\code\) over \(\FF_q\) of square class \(\chi\) has \(q^h M_\chi(n,k-h,q)\) vectors in \(\variety_\code \setminus \hull(\code)\), indeed \(\lvert \hull(\code) \rvert\) many for each non-zero self-orthogonal vector of its LCD component.

  The square class of \(\code(\mat{A})^\perp\), denoted by \(\chi_{\mat{A}^\perp}\), is \(\chi_{\mat{A}^\perp} = \chi \chi_q(-1)^h\) by~\cite[Prop.~5]{ITIT:LiShiLin25}, thus from \(\hull(\code(\mat{A})) = \hull(\code(\mat{A})^\perp)\) we get \(\lvert \variety_{\code(\mat{A})^\perp} \setminus \hull(\code(\mat{A})) \rvert = q^h M_{\chi \chi_q(-1)^h}(n,n-k-h,q)\) independently of \(\mat{A}\). This number is \(0\) if and only if the set \(\mathcal{X}_{n,k+1,q}^{\chi,h+1}\) is empty (see~\cite[Rem.~3]{localinv}), so our statement is vacuously true -- there is no \(\vec{y}\) to sample from \(\variety_{\code(\mat{A})^\perp} \setminus \hull(\code(\mat{A}))\) and there is no \(\dist_{n,k+1,q}^{\chi,h+1}\). Otherwise, the code \(\code([\mat{A} \Mid \vec{y}^\transpose])\) (as any \(\code\) of type \((h+1, \chi)\)) is by~\cite[Prop.~2]{ITIT:LiShiLin25} obtained from \((q^{h+1} - 1) \cdot q^{k}\) different \([n,k]\)-linear \(h\)-hollow codes over \(\FF_q\) with square class \(\chi\) by appending a self-orthogonal vector from its dual. Thus for any \(h\)-hollow code \(\code\) of square class \(\chi\) we have
  \begin{align*}
    \iftwocolumns{&}{}\probsub{\mat{A} \sample \dist_{n,k,q}^{\chi,h},\vec{y} \sample \variety_{\code(\mat{A})^\perp}\setminus\hull(\code(\mat{A}))}{\code = \code([\mat{A} \Mid \vec{y}^\transpose])} \iftwocolumns{\\}{}
    &= \prob{\mat{A} \sample \dist_{n,k,q}^{\chi,h}} \frac{(q^{h+1} - 1) \cdot q^k}{q^h M_{\chi \chi_q(-1)^h}(n,n-k-h,q)} \\
    &= \frac{1}{\left\lvert \mathcal{X}_{n,k,q}^{\chi,h} \right\rvert} \frac{(q^{h+1}-1) \cdot q^{k-h}}{M_{\chi \chi_q(-1)^h}(n,n-k-h,q)}
    \overset{\bullet}{=} \frac{1}{\left\lvert \mathcal{X}_{n,k+1,q}^{\chi,h+1} \right\rvert},
  \end{align*}
  where for \(\bullet\) we use~\cite[Thm.~5]{ITIT:LiShiLin25}, so the distribution of \(\code([\mat{A} \Mid \vec{y}^\transpose])\) is indeed that of a uniformly random \([n,k+1]\)-linear \((h+1)\)-hollow code over \(\FF_q\) with square class \(\chi\). Multiplying by \(\mat{U} \sample \GL_{k+1}(\FF_q)\) then translates this statement to generator matrices, that is, \(\mat{U} \cdot [\mat{A} \Mid \vec{y}^\transpose]\) is distributed according to \(\dist_{n,k+1,q}^{\chi,h+1}\).
\end{proof}

The algorithm we give will (contrary to the claims of~\cite{EC:AlbBenLai25}) sample from \(\dist_{n,k,q}^{\chi,h}\), where \(\chi = 1\) whenever \(h = k\), for any rate and hull dimension. Indeed the only exception is the self-dual parameter regime of \(k = \frac{n}{2}\), \(h = k\) and thus \(\chi = 1\), where according to~\cite{ISIT:Massey98} self-dual codes exist if and only if \(\chi_q(-1)^k = 1\). Our algorithm will naturally arrive at this edge-case and successfully finish sampling with overwhelming probability if and only if this criterion is satisfied. Hence our algorithm implies, to the best of our knowledge, the first algorithm that samples uniformly random self-dual codes.

\begin{definition}[Algorithm \(\SampleCode\)]\label{def:algo:samplecode}
  Let \(n, k, h, q, \chi \in \ZZ\) with \(0 < k \leq \frac{n}{2}\), \(0 \leq h \leq k\), \(\chi \in \set{\pm 1}\), where necessarily \(\chi = 1\) when \(h = k\), and \(q\) odd prime power. Consider the algorithm in \cref{fig:algo:samplecode} that on input \(n, k, h, q, \chi\) outputs an \(h\)-hollow matrix \(\mat{A} \in \FF_q^{k \times n}\) with \(\chi_\mat{A} = \chi\) or aborts with \(\bot\).
\end{definition}

\begin{figure}[t]
  \centering
  \begin{pchstack}[boxed]
    \procedure[mode=text]{$\SampleCode(n,k,h,q,\chi)$}{
      \(\mat{A}_{(0)} \coloneqq \bot\);\;
      \pcfor \(i \coloneqq 1,\dots,k\) \pcdo \\
      \t \(\mat{A}_{(0)} \from \SampleLCD(n,k-h,q,\chi)\) \\
      \t \pcif \(\mat{A}_{(0)} \neq \bot\)\pcthen \pcbreak \pcelse \\
      \t\t \pcif \(i = k\)\pcthen \pcreturn \(\bot\) \pcelse \pccontinue \\
      \pcfor \(i \coloneqq 1,\dots,h\) \pcdo \\
      \t \(\mat{G} \coloneqq \Dual(\mat{A}_{(i-1)})\);\;
      \pcfor \(j \coloneqq 1,\dots,k\) \pcdo \\
      \t\t \(\vec{y}_{i} \from \SampleSelfOrthogonal(\mat{G})\);\;
      \pcif \(\vec{y}_i = \bot\) \textbf{or} \(\vec{y}_{i} \in \hull(\code(\mat{A}_{(i-1)}))\)\pcthen \\
      \t\t\t \pcif \(j = k\)\pcthen \pcreturn \(\bot\) \pcelse \pccontinue \\
      \t\t \(\mat{A}_{(i)} \coloneqq [\mat{A}_{(i-1)} \Mid \vec{y}_{i}^\transpose]\);\;
      \pcbreak \\
      \(\mat{U} \sample \GL_k(\FF_q)\);\;
      \(\mat{A} \coloneqq \mat{U} \mat{A}_{(h)}\);\;
      \pcreturn \(\mat{A}\)
    }
  \end{pchstack}
  \caption{The algorithm \(\SampleCode\). The algorithm \(\SampleLCD\) is in \protect\cref{fig:algo:samplelcd}, the algorithm \(\SampleSelfOrthogonal\) is in \protect\cref{fig:algo:sso}, the algorithms \(\Dual\) and \(\Hull\) are in \protect\cref{fig:algo:dual,fig:algo:hull}.}
  \label{fig:algo:samplecode}
\end{figure}

\begin{theorem}\label{thm:algo:samplecode}
  Assume the notation of \cref{def:algo:samplecode} with both \(n,q = \poly[k]\). The algorithm \(\SampleCode\) runs in time \(\poly[n,k,q]\), it succeeds with probability \(1 - \negl[k]\) unless \(k = \frac{n}{2}\), \(h=k\) and \(\chi_q(-1)^k \neq 1\), in which case it always aborts, and its output distribution conditioned on not aborting is \(\dist_{n,k,q}^{\chi,h}\).
\end{theorem}

\begin{proof}
  We discuss efficiency and correctness, failure probability and uniformity of the output distribution.

  \paragraph*{Efficiency and correctness}
  First assume the algorithm \(\SampleCode\) successfully outputs a matrix \(\mat{A}\). If \(h < k\) then the algorithm outputs a full-rank \(k \times n\) matrix over \(\FF_q\), since we start with a full-rank \((k-h) \times n\) matrix and append \(h\) rows that increment the rank, otherwise we return \(\bot\). The starting matrix \(\mat{A}_{(0)} \in \FF_q^{(k-h) \times n}\) is \(0\)-hollow, otherwise we return \(\bot\), and each appended row \(\vec{y}_{i}^\transpose\) increments the hull dimension. Indeed, \(\vec{y}_{i}\) is by \cref{lem:algo:sso} orthogonal to itself and every vector in the code \(\code(\mat{A}_{(i-1)})\), since it belongs to the dual code of \(\code(\mat{A}_{(i-1)})\). Since we exclude vectors that are in \(\hull(\code(\mat{A}_{(i-1)})) \subseteq \variety_{\code(\mat{A}_{(i-1)})^\perp}\), the vector \(\vec{y}_i\) is not already in \(\code(\mat{A}_{(i-1)})\), thus adding it increments the rank. Clearly \(\vec{y}_{i}\) is also in the hull of \(\mat{A}_{(i)}\), hence also the hull dimension increments. Thus \(\mat{A} \in \FF_q^{k \times n}\) is \(h\)-hollow. Since \(\mat{A}_{(0)}\) is an LCD component of \(\code(\mat{A})\), we have that \(\chi_\mat{A} = \chi_{\mat{A}_{(0)}} = \chi\), thus \(\code(\mat{A}) \in \mathcal{X}_{n,k,q}^{\chi,h}\). The argument for \(h=k\) is identical by remarking that the zero code \(\set{\vec{0}^n}\) is LCD as its dual is \(\FF_q^n\) and spanned by \(\mat{I}_n\). Hence in the first iteration we have \(\mat{G} = \mat{I}_n\) and we sample a uniformly random self-orthogonal vector \(\vec{y}_1\) from \(\FF_q^n\). The condition that \(\vec{y}_i\) is not in the hull of \(\mat{A}_{(i-1)}\) is in this case equivalent to linear independence, that is, the rank increases at every iteration and we end up with an \([n,k]\)-linear self-orthogonal code. Every step of the algorithm is efficient by \cref{lem:algo:samplelcd,lem:algo:sso}.

  \paragraph*{Failure probability}
  The algorithm \(\SampleLCD\) on input \(n,k-h,q,\chi\) always succeeds if \(h=k\) and otherwise succeeds with probability \(1 - \negl[k-h]\) and outputs a (uniformly random) \(0\)-hollow matrix \(\mat{A}_{(0)} \in \FF_q^{(k-h) \times n}\) with \(\chi_{\mat{A}_{(0)}} = \chi\) by \cref{lem:algo:samplelcd}. Observe that by repeating (at most) \(k\) times, the success probability rises to \(1 - \negl[(k-h) \cdot k] = 1 - \negl[k]\) by the same argument as in \cref{lem:algo:samplelcd}.

  Next, \(\SampleSelfOrthogonal\) returns a uniformly random self-orthogonal vector \(\vec{y}_{i}\) in \(\code(\mat{A}_{(i-1)})^\perp\) except with negligible probability in \(n\) by \cref{lem:algo:sso}. If \(\vec{y}_{i}\) is already in \(\code(\mat{A}_{(i-1)})\) then it must by construction be in \(\hull(\code(\mat{A}_{(i-1)}))\) and we reject it as it does not produce a new linear code. Using \cref{lem:code:variety:nsols} for \(\bullet\), we have that
  \begin{align*}
    \iftwocolumns{&}{}\probsub{\vec{y}_i \sample \variety_{\code(\mat{A}_{(i-1)})^\perp}}{\vec{y}_i \in \hull(\code(\mat{A}_{(i-1)}))} \iftwocolumns{\\}{}
    &= \frac{\left\lvert \hull(\code(\mat{A}_{(i-1)})) \right\rvert}{\left\lvert \variety_{\code(\mat{A}_{(i-1)})^\perp} \right\rvert}
    \overset{\bullet}{\leq} \frac{q^{i-1}}{q^{n - (k-h+i-1) - 2}} \\
    &= q^{2 i - (n - (k-h))}
    \leq q^{h - (n-k)}.
  \end{align*}
  This leads us to the following two cases.

  If it is not the case that \(k = \frac{n}{2}\) and \(h = k\), then the above is further bounded from above by \(q^{-1}\) and after repeating at most \(k\) times by \(q^{-k}\) using the union bound, which is negligible in \(k\). The probability that we succeed in sampling the hull generators \(\vec{y}_1,\dots,\vec{y}_h\) is then bounded from below by the Weierstrass product inequality as \((1 - q^{-k})^h \geq 1 - \frac{h}{q^{k}} = 1 - \negl[k]\), since we need to be successful all \(h\) times, otherwise we abort with \(\bot\). Thus, the probability that the algorithm successfully outputs a code generator matrix (rather than \(\bot\)) is overwhelming in \(k\).

  The one remaining case is that of \emph{self-dual codes}, that is \(k = \frac{n}{2}\) and \(h = k\), where the above estimate collapses. Until the last iteration, the argument is identical to the former case; we must show that the last vector \(\vec{y}_k\) can be chosen with overwhelming probability. In the last iteration we have a \([2 k, k-1]\)-linear self-orthogonal code spanned by \(\mat{A}_{(k-1)}\) and compute a generator \(\mat{G}\) of its dual \([2 k, k+1]\)-linear code. Since \(\mat{G}\) is \((k-1)\)-hollow, its hull co-dimension is precisely \(2\), thus we may run \(\SampleSelfOrthogonal\) on it. By~\cite[Prop.~5]{ITIT:CMTQ19}. we have that \(\chi_\mat{G} = \chi_q(-1)^{k-1}\), hence by \cref{prop:plane:hyper} the code \(\code(\mat{G})\) contains self-orthogonal vectors not already in \(\hull(\code(\mat{A}_{(k-1)}))\) if and only if it is a hyperbolic plane with \(\chi_\mat{G} = \chi_q(-1)\), that is if and only if \(\chi_q(-1)^k = 1\). Recall that this is precisely when self-dual codes exist, thus we rightly always abort with \(\bot\) when they do not. Every self-orthogonal vector in \(\code(\mat{A}_{(k-1)})^\perp\) is of the form \(\vec{y} = \vec{h} + \vec{t}\) with \(\vec{h} \in \hull(\code(\mat{A}_{(k-1)})^\perp)\) and \(\vec{t}\) a self-orthogonal vector of the chosen hyperbolic plane. By \cref{prop:plane:hyper}, the chosen regular plane contains a non-zero self-orthogonal \(\vec{t}\), hence it contains the line spanned by \(\vec{t}\), which itself contains at least \(q\) self-orthogonal vectors. Thus
  \begin{align*}
    &\probsub{\vec{y}_k \sample \variety_{\code(\mat{A}_{(k-1)})^\perp}}{\vec{y}_k \in \hull(\code(\mat{A}_{(k-1)}))} \iftwocolumns{\\}{}
    \iftwocolumns{&}{}= \frac{\left\lvert \hull(\code(\mat{A}_{(k-1)})) \right\rvert}{\left\lvert \variety_{\code(\mat{A}_{(k-1)})^\perp} \right\rvert}
    \leq \frac{q^{k-1}}{q q^{k-1}}
    \leq q^{-1}.
  \end{align*}
  We can now repeat the argument from above: after at most \(k\) repetitions, the probability of \(\SampleSelfOrthogonal\) failing by sampling a known hull vector is bounded by \(q^{-k}\) from above, which is negligible in \(k\). Hence the probability that we succeed in sampling the self-orthogonal code generators \(\vec{y}_1,\dots,\vec{y}_k\) is bounded from below by the Weierstrass product inequality as \((1 - q^{-k})^k \geq 1 - \frac{k}{q^{k}} = 1 - \negl[k]\), since we need to be successful all \(k\) times, otherwise we abort with \(\bot\).

  \paragraph*{Uniform sampling}
  Conditioned on \(\SampleCode\) succeeding, we simply apply \cref{lem:hollow:unif} iteratively \(h\) times. Indeed, since \(\mat{A}_{(0)}\) is distributed according to \(\dist_{n,k-h,q}^{\chi,0}\) by \cref{lem:algo:samplelcd} (otherwise we abort), we have that \(\mat{A}_{(1)}\) is distributed according to \(\dist_{n,k-h+1,q}^{\chi,1}\) (up to re-randomisation by a random basis change, which we apply at the end). Further, \(\mat{A}_{(i)}\) is distributed according to \(\dist_{n,k-h+i,q}^{\chi,i}\), so the output \(\mat{A} = \mat{U} \mat{A}_{(h)}\) for \(\mat{U} \sample \GL_k(\FF_q)\) is distributed according to \(\dist_{n,k,q}^{\chi,h}\).
\end{proof}

Using the fact that square classes of random hollow codes behave essentially like random bits~\cite[Thm.~1]{localinv}, we directly obtain an algorithm whose output distribution conditioned on succeeding is negligibly close in \((k-h) \cdot \log(q)\) to \(\dist_{n,k,q}^{h}\).

\begin{definition}[Algorithm \(\SampleHollow\)]\label{def:algo:samplehollow}
  Let \(n,k,h,q \in \ZZ\) with \(0 < k \leq \frac{n}{2}\), \(0 \leq h \leq k\) and \(q\) odd prime power. Consider the algorithm in \cref{fig:algo:samplehollow} that on input \(n,k,h,q\) outputs an \(h\)-hollow matrix \(\mat{A} \in \FF_q^{k \times n}\) or aborts with \(\bot\).
\end{definition}

\begin{figure}[t]
  \centering
  \begin{pchstack}[boxed]
    \procedure[mode=text]{$\SampleHollow(n,k,h,q)$}{
      \pcif \(h = k\)\pcthen \pcreturn \(\mat{A} \from \SampleCode(n,k,h,q,1)\) \\
      \pcelse \pcreturn \(\mat{A} \from \SampleCode(n,k,h,q,\chi \sample \set{\pm 1})\)
    }
  \end{pchstack}
  \caption{The algorithm \(\SampleHollow\) that, given \(n,k,h,q \in \ZZ\), outputs an \(h\)-hollow matrix \(\mat{A} \in \FF_q^{k \times n}\). The algorithm \(\SampleCode\) is in \protect\cref{fig:algo:samplecode}.}
  \label{fig:algo:samplehollow}
\end{figure}

\begin{corollary}\label{cor:algo:samplehollow}
  Assume the notation of \cref{def:algo:samplehollow} with both \(n,q = \poly[k]\). The algorithm \(\SampleHollow\) runs in time \(\poly[n,k,q]\), succeeds with probability \(1 - \negl[k]\) unless \(k = \frac{n}{2}\), \(h=k\) and \(\chi_q(-1)^k \neq 1\) in which case it always aborts, and its output distribution conditioned on not aborting \(\dist\) satisfies \(\Delta(\dist, \dist_{n,k,q}^h) \leq q^{-\frac{k-h}{2}}\) unless \(h = k\) in which case \(\dist = \dist_{n,k,q}^k\).
\end{corollary}

\begin{proof}
  Clearly the algorithm \(\SampleHollow\) inherits correctness, efficiency and success probability from the algorithm \(\SampleCode\) in \cref{def:algo:samplecode}. The distributions \(\dist\) and \(\dist_{n,k,q}^{h}\) clearly have the same support that is the set of all generator matrices of codes in \(\mathcal{X}_{n,k,q}^{h}\). The case \(h=k\) is immediate. Otherwise, for the total variation distance \(\Delta\) we have
  \begin{align*}
    \iftwocolumns{&}{}2 \Delta(\dist,\dist_{n,k,q}^h) \iftwocolumns{\\}{}
    &= \sum_{\condset{\mat{A}}{\code(\mat{A}) \in \mathcal{X}_{n,k,q}^h}} \left\lvert \prob{\mat{A} \sample \dist} - \prob{\mat{A} \sample \dist_{n,k,q}^h} \right\rvert \\
    &= \left\lvert \mathcal{X}_{n,k,q}^{+1,h} \right\rvert \left\lvert \left( \frac{1}{2 \left\lvert \mathcal{X}_{n,k,q}^{+1,h} \right\rvert} - \frac{1}{\left\lvert \mathcal{X}_{n,k,q}^h \right\rvert} \right) \right\rvert \iftwocolumns{\\}{}
    \iftwocolumns{&}{}+ \left\lvert \mathcal{X}_{n,k,q}^{-1,h} \right\rvert \left\lvert \left( \frac{1}{2 \left\lvert \mathcal{X}_{n,k,q}^{-1,h} \right\rvert} - \frac{1}{\left\lvert \mathcal{X}_{n,k,q}^h \right\rvert} \right) \right\rvert \\
    &= \frac{\left\lvert \left\lvert \mathcal{X}_{n,k,q}^{+1,h} \right\rvert - \left\lvert \mathcal{X}_{n,k,q}^{-1,h} \right\rvert \right\rvert}{\left\lvert \mathcal{X}_{n,k,q}^h \right\rvert}
    \overset{\bullet}{\leq} 2 q^{-\frac{k-h}{2}},
  \end{align*}
  where for \(\bullet\) we use~\cite[Thm.~1]{localinv}.
\end{proof}

For \(h\) close to \(k\) this can of course be improved by replacing \(\chi \sample \set{\pm 1}\) by a better estimate, for example by working with the mass formulae in~\cite{ITIT:LiShiLin25} for the desired parameters directly.

\section{On Total Orthogonality over Finite Extension Fields}\label{sec:total}

The algorithm from \cref{def:algo:samplecode} allows us to sample codes of any hull dimension over fields of any odd prime power size. However, the notion of a (Euclidean) hull may be insufficient over non-prime fields in certain contexts, for example when considering permutation equivalence, as noticed already in~\cite[Sec.~IV.A]{ISIT:BarOtmSae19} and later in~\cite[Sec.~6]{upke2}. Indeed, each automorphism of the field induces its own `inner product' and a separate notion of orthogonality and, by extension, hull. In this section and the next, we explore a much stricter variant of orthogonality we call \emph{total orthogonality}, meaning orthogonality with respect to every one of the degree-many sesquilinear forms induced by field automorphisms simultaneously. This section is reserved for characterising total orthogonality of vectors, while the next section shows how to sample `random' codes with a particularly large \emph{total hull}, that is a subcode that acts as the hull in every geometry.

\subsection{Total Orthogonality}\label{sec:total:defs}

It will prove useful to think of \(\FF_q\) for \(q = p^d\) as a \(d\)-dimensional vector space over \(\FF_p\). In particular, for a fixed field generator \(z\) of \(\FF_q/\FF_p\) we will think of a vector \(\vec{x} \in \FF_q^n\) as giving \(d\) vectors by decomposing each entry according to the basis \(\mathcal{B}_z\) and considering vectors over \(\FF_p\) given by coefficients at a single degree.

\begin{definition}[Coefficient projection]\label{def:coefproj}
  Denote by \(\pi_i^{(f,z)} \colon \FF_q \to \FF_p\) for \(i \in \ZZ_d\) the \(i\)-th \emph{coefficient projection} that maps \(\pi_i^{(f, z)} \colon \sum_{j=0}^{d-1} a_j z^j \mapsto a_i\). When the polynomial \(f\) and the root \(z\) are fixed, we will often omit the superscript. We extend these maps to vectors \(\vec{x} \in \FF_q^n\) coordinate-wise, that is \(\pi_i(\vec{x}) = [\pi_i(x_j)]_{j=1}^n\), and likewise to matrices.

  Denote by \(\pi = (\pi_0,\dots,\pi_{d-1}) \colon \FF_q \to \FF_p^d\) the \(\FF_p\)-vector space isomorphism that maps \(\sum_{i=0}^{d-1} a_i z^i \mapsto [a_i]_{i=0}^{d-1}\) (cf.~\emph{coefficient embedding}), and by \(\pi^{-1}\) its inverse we call the \emph{coefficient lift}. We extend these maps to vectors and matrices, that is \(\pi^{-1}\) maps a matrix from \(\FF_p^{d \times n}\) to a vector in \(\FF_q^n\).
\end{definition}

Besides the usual Euclidean inner product, recall from \cref{def:galois:inner} that over an extension field \(\FF_q/\FF_p\) there are degree-many Galois inner products on \(\FF_q^n\), one for each automorphism, where the Euclidean inner product corresponds to the identity map (the \(0\)-th automorphism).

\begin{definition}[Galois and total orthogonality of vectors]\label{def:torthogonality:vectors}
  Vectors \(\vec{x},\vec{y} \in \FF_q^n\) are \emph{\(\phi_i\)-orthogonal} if \(\inner{\vec{x}}{\vec{y}}_{\phi_i} = 0\), sometimes denoted as \(\vec{x} \perp_{\phi_i} \vec{y}\). Two vectors that are \(\phi_i\)-orthogonal for every \(i \in \ZZ_d\) will be called \emph{totally orthogonal}. A vector is \emph{\(\phi_i\)-self-orthogonal} if it is \(\phi_i\)-orthogonal to itself. We call a vector \emph{totally self-orthogonal}, if it is \(\phi_i\)-self-orthogonal for every \(i \in \ZZ_d\).
\end{definition}

Let us check that orthogonality between vectors translates to orthogonality between the spaces they span for each Galois inner product, despite Galois inner products not being \(\FF_q\)-linear in the second component in general (see \cref{prop:galois:properties}).

\begin{proposition}\label{prop:torthogonality:bases}
  Let \(\vec{x}_1,\dots,\vec{x}_t,\vec{y}_1,\dots,\vec{y}_s \in \FF_q^n\), and \(r \in \ZZ_d\) for \(q = p^d\). Then we have that for every \(i = 1,\dots,t\) and \(j = 1,\dots,s\) it holds \(\inner{\vec{x}_i}{\vec{y}_j}_{\phi_r} = 0\) if and only if for every \(\vec{x} \in \code(\set{\vec{x}_i}_{i=1}^t)\) and every \(\vec{y} \in \code(\set{\vec{y}_j}_{j=1}^s)\) it holds \(\inner{\vec{x}}{\vec{y}}_{\phi_r} = 0\).
\end{proposition}

\begin{proof}
  The if direction is clear, let us show the only if direction. It is clear from \(\FF_q\)-linearity in the first argument of \(\inner{\cdot}{\cdot}_{\phi_r}\) that if \(\inner{\vec{x}_i}{\vec{y}}_{\phi_r} = 0\) for \(i=1,\dots,t\), then also \(\code(\set{\vec{x}_i}_{i=1}^t) \perp_{\phi_r} \vec{y}_j\) for every \(j = 1,\dots,s\). For the second argument, where we only have \(\FF_p\)-linearity in general, observe that if \(\inner{\vec{x}}{\vec{y}_i}_{\phi_r} = 0\) for \(i=1,\dots,k\) then
  \[
    \inner{\vec{x}}{\sum_{i=1}^{k} \alpha_i \vec{y}_i}_{\phi_r}
    = \sum_{i=1}^{k} \phi_r(\alpha_i) \inner{\vec{x}}{\vec{y}_i}_{\phi_r}
    = 0.
  \]
  Thus \(\phi_r\)-orthogonality of bases translates to \(\phi_r\)-orthogonality of their spans.
\end{proof}

Each of the Galois inner products endows \(\FF_q^n\) with a different geometry, and in particular changes the notion of orthogonality. By \cref{prop:torthogonality:bases}, it makes sense to speak of Galois duals and hulls in terms of their generator matrices.

\begin{definition}[Galois dual and hull {\cite[IV.A]{ISIT:BarOtmSae19}}]\label{def:galois:hull}
  For a linear code \(\code \leq \FF_q^n\) and \(i \in \ZZ_d\) we define its \(i\)-th Galois dual as
  \[
    \code^{\perp,\phi_i} = \condset{\vec{y} \in \FF_q^n}{\forall \vec{x} \in \code \colon \inner{\vec{x}}{\vec{y}}_{\phi_i} = 0},
  \]
  and its \(i\)-th Galois hull as
  \[
    \hull_{\phi_i}(\code) = \condset{\vec{y} \in \code}{\forall \vec{x} \in \code \colon \inner{\vec{x}}{\vec{y}}_{\phi_i} = 0} = \code \cap \code^{\perp,\phi_i}.
  \]
  We give explicit algorithms to compute them in \cref{fig:galois:hull}.
\end{definition}

\begin{figure}[t]
  \centering
  \begin{pcvstack}[boxed]
    \procedure[mode=text]{$\DualGalois(\mat{A}, i)$}{
      \(\mat{A'} \coloneqq \phi_{-i}(\mat{A})\);\;
      \pcreturn \(\Dual(\mat{A'})\)
    }
    \pcvspace
    \procedure[mode=text]{$\HullGalois(\mat{A}, i)$}{
      \(\mat{A}^\perp \coloneqq \Dual(\mat{A})\);\;
      \(\mat{M} \coloneqq [\mat{A}^\perp \Mid \phi_{-i}(\mat{A})]\) \\
      \(\mat{H} \coloneqq\) basis of right kernel of \(\mat{M}\);\;
      \pcreturn \(\mat{H}\)
    }
  \end{pcvstack}
  \caption{The \(i\)-th Galois dual and hull of a code.}
  \label{fig:galois:hull}
\end{figure}

\begin{definition}[Galois and total orthogonality of codes]\label{def:totalort}
  Two linear codes \(\code_1, \code_2 \leq \FF_q^n\) are \(\phi_i\)-orthogonal if \(\code_1 \subseteq \code_2^{\perp,\phi_i}\), or equivalently if \(\code_2 \subseteq \code_1^{\perp,\phi_{d-i}}\). A linear code \(\code\) is \(\phi_i\)-self-orthogonal if \(\code \subseteq \code^{\perp,\phi_i}\), and further totally self-orthogonal if this holds for every \(i \in \ZZ_d\).
\end{definition}

\begin{remark}\label{rem:totalort}
  Note that \cref{def:totalort} of total (self-)orthogonality can be seen as a slight generalisation of~\cite[Def.~1.2.1]{NebeRainsSloane2006}, if we considered codes more broadly over rings. In~\cite[Def.~1.2.1]{NebeRainsSloane2006} codes over rings are endowed with a set of bilinear forms \(M\), and orthogonality is considered with respect to all forms in \(M\), much like our total orthogonality. The difference is that our forms \(\inner{\cdot}{\cdot}_{\phi_i}\) are generally not bilinear but \(\phi_i\)-sesquilinear.
\end{remark}

\begin{definition}[Total hull]\label{def:totalhull}
  For a linear code \(\code\) over \(\FF_{p^d}\) define its \emph{total hull} as the vector subspace
  \[
    \hull^\star(\code) = \bigcap_{i \in \ZZ_d} \hull_{\phi_i}(\code).
  \]
\end{definition}

\noindent Over prime fields \(\hull^\star\) coincides with the Euclidean hull.

\subsection{The Coefficient Code and the Frobenius Closure}\label{sec:total:codes}

We start by defining two codes that are naturally associated to a vector \(\vec{x} \in \FF_q^n\), and will later assist us in characterising total orthogonality.

\begin{definition}[Coefficient code]\label{def:coefcode}
  Let \(q = p^d\) for \(p\) prime and \(d \in \NN\), and fix \(f \in \FF_p[X]\) irreducible of degree \(d\) and \(z \in \variety_f\). For \(\vec{x} \in \FF_q^n\) define its \emph{coefficient code} as
  \[
    \Pi^{(f,z)}(\vec{x}) = \Span_{\FF_p}\left(\condset{\pi_i^{(f,z)}(\vec{x})}{i \in \ZZ_d}\right).
  \]
  It is defined over \(\FF_p\), has length \(n\) and dimension at most \(d\).
\end{definition}

Since different pairs \((f, z)\) and \((g, w)\) of an irreducible polynomial over \(\FF_p[X]\) of degree \(d\) and its root in \(\FF_q\) correspond to an \(\FF_p\)-basis change, the choice of the field generator does not matter and we omit the superscript.

\begin{definition}[Moore code]\label{def:moorecode}
  Let \(q = p^d\) and \(\vec{x} \in \FF_q^n\). The \emph{Moore code} of \(\vec{x}\) is defined as
  \[
    \Phi(\vec{x}) = \Span_{\FF_q}\left( \condset{\phi_i(\vec{x})}{i \in \ZZ_d} \right).
  \]
  It is defined over \(\FF_q\), has length \(n\) and dimension at most \(d\).
\end{definition}

The Moore code is essentially the Frobenius closure of a single vector, so called because it is spanned by the Moore matrix of the vector's entries.

\begin{proposition}
  The following are equivalent: (i) the code \(\Phi(\vec{x})\) has dimension \(d\) over \(\FF_q\), (ii) the code \(\Pi(\vec{x})\) has dimension \(d\) over \(\FF_p\), and (iii) the entries of \(\vec{x}\) span \(\FF_q\) as a \(\FF_p\)-vector space.
\end{proposition}

\begin{proof}
  The code \(\Phi(\vec{x})\) is generated by the matrix \(\mat{M}_d(\vec{x})^\transpose\) where \(\vec{x}\) is viewed as a list of \(n\) elements in \(\FF_q\). By \cref{def:moore}, the determinant of \(\mat{M}_d(a_1,\dots,a_d)\) is non-zero if and only if the elements \(a_1,\dots,a_d \in \FF_q\) are \(\FF_p\)-linearly independent, that is they form a basis for the \(\FF_p\)-vector space \(\FF_q \iso \FF_p^d\). This immediately implies that if \(n \geq d\) and \(\vec{x} \in \FF_q^n\) contains at least \(d\) \(\FF_p\)-linearly independent entries, which necessarily span \(\FF_q\) as an \(\FF_p\)-vector space, then \(\Phi(\vec{x})\) is full-rank and thus an \([n,d]\)-linear code over \(\FF_q\). The entries of \(\vec{x}\) span \(\FF_p^d\) precisely when the coefficient projection matrix \(\pi(\vec{x})\) spanning \(\Pi(\mat{x})\) has full rank \(d\). When \(n < d\) all statements are false, completing the proof.
\end{proof}

\subsection{Characterisations of Total (Self-)Orthogonality}\label{sec:total:char}

Fix an irreducible polynomial \(f\) of degree \(d\) over \(\FF_p\) and its root \(z \in \FF_q\). Let us first compute what it means for two vectors to be \(\phi_r\)-orthogonal for \(r \in \ZZ_d\) in terms of their coefficients with respect to the basis \(\mathcal{B}_z\).

\begin{proposition}\label{prop:orthogonal}
  For every two vectors \(\vec{x},\vec{y} \in \FF_q^n\) with \(q = p^d\) for \(p\) prime and \(d,n \in \NN\), and every \(r \in \ZZ_d\), we have
  \[
    \inner{\vec{x}}{\vec{y}}_{\phi_r}
    = \sum_{i=0}^{d-1} \sum_{j=0}^{d-1} z^{i + j p^r} \inner{\pi_i(\vec{x})}{\pi_j(\vec{y})}.
  \]
\end{proposition}

\begin{proof}
  Write vectors \(\vec{x}, \vec{y} \in \FF_q^n\) as
  \[
    \vec{x} = \left[\sum_{i=0}^{d-1} a_i^{(t)} z^i \right]_{t=1}^n
    \qquad
    \text{and}
    \qquad
    \vec{y} = \left[\sum_{j=0}^{d-1} b_j^{(t)} z^j \right]_{t=1}^n,
  \]
  and compute
  \begin{align*}
    \inner{\vec{x}}{\vec{y}}_{\phi_r}
    &= \sum_{t=1}^{n} \left( \sum_{i=0}^{d-1} a_i^{(t)} z^i \right) \cdot \left( \sum_{j=0}^{d-1} b_j^{(t)} \phi_r(z^j) \right) \iftwocolumns{\\}{}
    \iftwocolumns{&}{}= \sum_{t=1}^{n} \sum_{i=0}^{d-1} \sum_{j=0}^{d-1} a_i^{(t)} b_j^{(t)} z^{i + j p^r} \\
    &= \sum_{i=0}^{d-1} \sum_{j=0}^{d-1} z^{i + j p^r} \left( \sum_{t=1}^{n} a_i^{(t)} b_j^{(t)} \right) \iftwocolumns{\\}{}
    \iftwocolumns{&}{}= \sum_{i=0}^{d-1} \sum_{j=0}^{d-1} z^{i + j p^r} \inner{\pi_i(\vec{x})}{\pi_j(\vec{y})}. \qedhere
  \end{align*}
\end{proof}

We may ignore how the generator wraps around by simply setting each coefficient in the sum in \cref{prop:orthogonal} to \(0\), hence the following lemma is immediate.

\begin{lemma}\label{lem:orthogonal}
  A sufficient condition that the vectors \(\vec{x}\) and \(\vec{y}\) are totally orthogonal is that \(\Pi(\vec{x})\) and \(\Pi(\vec{y})\) are orthogonal codes over \(\FF_p\) (that is one is in the dual of the other). In particular, if \(\Pi(\vec{x})\) is a self-orthogonal code over \(\FF_p\), then \(\vec{x} \in \FF_q^n\) is totally self-orthogonal.
\end{lemma}

We can imagine that for a polynomial \(f\) with many non-zero coefficients, the generator \(z\) wraps around in a way that creates a condition on non-trivial sums of \(\inner{\pi_i(\vec{x})}{\pi_j(\vec{y})}\) being \(0\), even when some \(\inner{\pi_i(\vec{x})}{\pi_j(\vec{y})} \neq 0\), complicating the analysis of total orthogonality. Yet, the above sufficient condition is also necessary.

\begin{theorem}\label{thm:orthogonal}
  Two vectors \(\vec{x}, \vec{y} \in \FF_q^n\) are totally orthogonal if and only if \(\Pi(\vec{x})\) and \(\Pi(\vec{y})\) are orthogonal codes over \(\FF_p\). In particular, \(\vec{x} \in \FF_q\) is totally self-orthogonal if and only if \(\Pi(\vec{x})\) is a self-orthogonal code over \(\FF_p\).
\end{theorem}

\begin{proof}
  Let us write \(c_{i,j} = \inner{\pi_i(\vec{x})}{\pi_j(\vec{y})}\) for brevity. Observe that setting the \(\phi_r\)-orthogonality sums from \cref{prop:orthogonal} for \(r \in \ZZ_d\) to \(0\) yields a homogeneous system of \(d^2\) linear equations in the \(d^2\) variables \(c_{i,j}\), one equation per automorphism and per degree.

  Concretely, the system is the following. Order the pairs \((i,j) \in \ZZ_d^2\) and the pairs \((r, m) \in \ZZ_d^2\) lexicographically, these will number our columns and rows, respectively. Then the system of equations we are considering is the following:
  \begin{align*}
    \mat{A} \vec{c} =
    \begin{bmatrix}
      \pi_m\left(z^{i + j p^r}\right)
    \end{bmatrix}_{\genfrac{}{}{0pt}{}{(r,m) \in \ZZ_d^2}{(i,j) \in \ZZ_d^2}}
    \begin{bmatrix}
      c_{i,j}
    \end{bmatrix}_{(i,j) \in \ZZ_d^2}^\transpose
    =
    \vec{0}^{d^2}.
  \end{align*}
  In words, we have \(d\) collections of \(d\) successive rows, the \(r\)-th collection corresponding to the automorphism \(\phi_r\) for \(r \in \ZZ_d\). Within each collection, the \(m\)-th row corresponds to the coefficient of degree \(m \in \ZZ_d\), we call these the degree-\(m\) rows. Our goal is to prove that the \(d^2 \times d^2\) matrix \(\mat{A}\) in the above equation is full-rank, meaning the only solution is \(c_{i,j} = 0\) for all \((i,j) \in \ZZ_d^2\). We use the following trick.

  Instead of considering \(\mat{A}\) directly, consider the non-projected \(d \times d^2\) matrix \(\mat{A'} = [z^{i + j p^r}]_{r \in \ZZ_d,(i,j) \in \ZZ_d^2}\). Each entry in \(\mat{A'}\) corresponds to a \(d\)-entry sub-column in \(\mat{A}\). Observe that performing rank-preserving row operations on \(\mat{A'}\) corresponds to rank-preserving row operations \(\mat{A}\), that is they factor through the coefficient projection. Indeed, adding a row to another in \(\mat{A'}\) corresponds to adding \(d\) rows to another \(d\) rows in \(\mat{A}\), and swapping two rows of \(\mat{A'}\) corresponds to swapping two collections of \(d\) rows of \(\mat{A}\). The last ingredient is that multiplying a row of \(\mat{A'}\) with a non-zero field element \(a \in \FF_q\) corresponds to a basis change when viewing \(\FF_q\) as a \(\FF_p\) vector space, thus it corresponds to multiplying a matrix \(\mat{U}_a \in \GL_d(\FF_p)\) to the matching \(d\) rows of \(\mat{A}\).

  Now we compute
  \begin{align*}
    \iftwocolumns{&}{}\begin{bmatrix}
      z^{i + j p^r}
    \end{bmatrix}_{(i,j) \in \ZZ_d^2, r \in \ZZ_d}
    \iftwocolumns{}{&}=
    \begin{bmatrix}
      z^i \phi_r(z^j)
    \end{bmatrix}_{(i,j) \in \ZZ_d^2, r \in \ZZ_d} \\
    &=
    \begin{bmatrix}
      \begin{bmatrix}
        \phi_0(z^j)
      \end{bmatrix}_{j \in \ZZ_d}
      & \dots &
      \begin{bmatrix}
        \phi_0(z^j)
      \end{bmatrix}_{j \in \ZZ_d} \\
      \vdots & \ddots & \vdots \\
      \begin{bmatrix}
        \phi_{d-1}(z^j)
      \end{bmatrix}_{j \in \ZZ_d}
      & \dots &
      \begin{bmatrix}
        \phi_{d-1}(z^j)
      \end{bmatrix}_{j \in \ZZ_d}
    \end{bmatrix} \iftwocolumns{\\}{}
    \iftwocolumns{&\;\;}{}\cdot
    \begin{bmatrix}
      z^0 \mat{I}_d & & \\
      & \ddots & \\
      & & z^{d-1} \mat{I}_d
    \end{bmatrix} \\
    &=
    \begin{bmatrix}
      \mat{M}_d(z^0,\dots,z^{d-1})^\transpose & \dots & \mat{M}_d(z^0,\dots,z^{d-1})^\transpose
    \end{bmatrix} \iftwocolumns{\\}{}
    \iftwocolumns{&\;\;}{}\cdot
    \begin{bmatrix}
      z^0 \mat{I}_d & & \\
      & \ddots & \\
      & & z^{d-1} \mat{I}_d
    \end{bmatrix}
  \end{align*}
  Now recall that the Moore matrix \(\mat{M}_d(z^0,\dots,z^{d-1})\) is full-rank since the elements \(\set{z^0,\dots,z^{d-1}}\) form an \(\FF_p\)-basis for \(\FF_q\) as a vector space. Thus the above matrix can be transformed with row operations (by multiplying from the left with \(\mat{M}_d(z^0,\dots,z^{d-1})^{-\transpose}\)) into
  \begin{align*}
    &\sim
    \begin{bmatrix}
      \mat{I}_d & \dots & \mat{I}_d
    \end{bmatrix}
    \cdot
    \begin{bmatrix}
      z^0 \mat{I}_d & & \\
      & \ddots & \\
      & & z^{d-1} \mat{I}_d
    \end{bmatrix} \iftwocolumns{\\}{}
    \iftwocolumns{&}{}=
    \begin{bmatrix}
      z^0 \mat{I}_d & \dots & z^{d-1} \mat{I}_d
    \end{bmatrix}.
  \end{align*}
  Projecting this matrix back to a \(d^2 \times d^2\) matrix over \(\FF_p\) using coefficient embeddings yields a matrix that is clearly full-rank as it has precisely one \(1\) per row and per column, thus its determinant is \(\pm 1\). Hence \(\mat{A}\) is full-rank.
\end{proof}

\begin{remark}\label{rem:orthogonal}
  What may be surprising about \cref{thm:orthogonal} is that the crude method of setting every coefficient of the sum in \cref{prop:orthogonal} to \(0\) is the best we can do. In particular, this method \emph{ignores the multiplication structure} of \(\FF_q\), described by how the generator wraps around, and only treats \(\FF_q \iso \FF_p^d\) as an \(\FF_p\)-vector space.
\end{remark}

\cref{thm:orthogonal} gives us a precise characterisation of total orthogonality of vectors in \(\FF_q\) based on their coefficient codes over the base field \(\FF_p\). We also characterise total orthogonality of vectors using their Moore codes over the extension field \(\FF_q\), but we did not find the result as useful.

\begin{lemma}
  Let \(q = p^d\) be any prime power. The following are equivalent: (i) the vectors \(\vec{x}, \vec{y} \in \FF_q^n\) are totally orthogonal, (ii) the codes \(\Phi(\vec{x}), \Phi(\vec{y}) \leq \FF_q^n\) are totally orthogonal, and (iii) the codes \(\Phi(\vec{x}), \Phi(\vec{y}) \leq \FF_q^n\) are Euclidean orthogonal.
  In particular, \(\vec{x} \in \FF_q^n\) is totally self-orthogonal if and only if \(\Phi(\vec{x}) \leq \FF_q^n\) is totally self-orthogonal, equivalently Euclidean self-orthogonal.
\end{lemma}

\begin{proof}
  Suppose \(\inner{\vec{x}}{\vec{y}}_{\phi_i} = \inner{\vec{x}}{\phi_i(\vec{y})} = 0\) for all \(i \in \ZZ_d\). Then also
  \begin{align*}
    \inner{\phi_i(\vec{x})}{\phi_j(\vec{y})}_{\phi_r}
    &= \inner{\phi_i(\vec{x})}{\phi_{r+j}(\vec{y})}
    = \phi_i(\inner{\vec{x}}{\phi_{r+j-i}(\vec{y})}) \iftwocolumns{\\}{}
    \iftwocolumns{&}{}= \phi_i(0)
    = 0.
  \end{align*}
  The converse is trivial. The equivalence of total orthogonality and Euclidean orthogonality of \(\Phi(\vec{x})\) and \(\Phi(\vec{y})\) follows from conjugate symmetry and the fact that \(\Phi(\vec{x})\) and \(\Phi(\vec{y})\) are closed under action with any automorphism \(\phi_i\).
\end{proof}

\begin{remark}
  \cref{thm:orthogonal} can now be restated as
  \[
    \Pi(\vec{x}) \perp \Pi(\vec{y})
    \iff
    \Phi(\vec{x}) \perp_\star \Phi(\vec{y}).
  \]
\end{remark}

\section{How To Sample a Total Hull}\label{sec:totalhollow}

In this section, we give an algorithm to generate a \([n,k]\)-linear code over \(\FF_q\) that is \emph{totally \(h\)-hollow}, that is has a total hull of dimension \(h\), provided there is an algorithm to sample \(d h\)-hollow codes over \(\FF_p\). The cost we incur is non-trivial, however, since the rate \(r\) of codes our algorithm can output is bounded by \(\frac{1}{2 d}\) from above.

\subsection{Our Algorithm}

Our algorithm is based on the characterisation of total orthogonality in terms of coefficient codes from \cref{thm:orthogonal}. By the results of the previous section, it suffices to generate \(k\) linear codes \(\code_i\) over \(\FF_p\) of dimension \(d\) for \(i=1,\dots,k\) and lift them to vectors \(\vec{x}_i\) over \(\FF_q\) using \(\pi^{-1}\), such that \(\code_i \perp \code_j\) for all \(i=1,\dots,h\) and \(j=1,\dots,k\). If \(\code_i\) are jointly linearly independent, this will yield \(k\) vectors \(\vec{x}_i\) over \(\FF_q\) for \(i=1,\dots,k\) generating a code \(\code\), such that the first \(h\) of them generate a subspace \(\mathfrak{H} \leq \code\) that is \(\phi_r\)-orthogonal to \(\code\) for every \(r \in \ZZ_d\) and thus totally orthogonal to \(\code\), hence \(\dim(\mathfrak{H}) = h\).

\begin{definition}[Algorithm \(\SampleTotal\)]\label{def:algo:sampletotal}
  Let \(n,k,h,p,q,\chi \in \ZZ\) with \(0 \leq h \leq k \leq \frac{n}{2 d}\), \(\chi \in \set{\pm 1}\) (where \(\chi = 1\) when \(h=k\)), \(p\) odd prime and \(q = p^d\) for \(d \in \NN\). Consider the algorithm in \cref{fig:algo:sampletotal} that on input \(n,k,h,p,d,\chi\) outputs a generator matrix for a totally \(h\)-hollow \([n,k]\)-linear code over \(\FF_q\) with Euclidean square class \(\chi\) or aborts with \(\bot\).
\end{definition}

\begin{figure}[t]
  \centering
  \begin{pchstack}[boxed]
    \procedure[mode=text]{$\SampleTotal(n,k,h,p,d,\chi)$}{
      \pcassert \(d k \leq \frac{n}{2}\) \pcand \(h \leq k\);\;
      \pcfor \(i \coloneqq 1,\dots,k\) \pcdo \\
      \t \(\mat{T} \from \SampleLCD(n, k - h, p^d, \chi)\) \\
      \t \pcif \(\mat{T} \neq \bot\) \pcthen \pcbreak \pcelse \\
      \t\t \pcif \(i=k\)\pcthen \pcreturn \(\bot\) \pcelse \pccontinue \\
      \(\mat{A}_{(0)} \coloneqq [\pi(\mat{T}[1]) \Mid \cdots \Mid \pi(\mat{T}[k-h])]\) ;\;
      \(\mat{H} \coloneqq [\cdot]\) \\
      \pcfor \(i \coloneqq 1,\dots,d h\) \pcdo \\
      \t \(\mat{G} \coloneqq \Dual(\mat{A}_{(i-1)})\);\;
      \pcfor \(j \coloneqq 1,\dots,d k\) \pcdo \\
      \t\t \(\vec{y}_{i} \from \SampleSelfOrthogonal(\mat{G})\);\;
      \pcif \(\vec{y}_i = \bot\) \textbf{or} \(\vec{y}_{i} \in \hull(\code(\mat{A}_{(i-1)}))\)\pcthen \\
      \t\t\t \pcif \(j = d k\)\pcthen \pcreturn \(\bot\) \pcelse \pccontinue \\
      \t\t \(\mat{A}_{(i)} \coloneqq [\mat{A}_{(i-1)} \Mid \vec{y}_{i}^\transpose]\);\;
      \(\mat{H} \coloneqq [\mat{H} \Mid \vec{y}_{i}^\transpose]\);\;
      \pcbreak \\
      \(\mat{A} \coloneqq [\cdot]\);\;
      \(\mat{U} \sample \GL_{k}(\FF_q)\) \\
      \pcfor \(i \coloneqq 1,\dots,h\) \pcdo\, \(\mat{A}[i] \coloneqq \pi^{-1}(\mat{H}\llbracket d \cdot (i-1) + 1, d i \rrbracket)\) \\
      \(\mat{A}\llbracket h+1, k \rrbracket \coloneqq \mat{T}\);\;
      \pcreturn \(\mat{U} \mat{A}\)
    }
  \end{pchstack}
  \caption{The algorithm \(\SampleTotal\). The algorithm \(\SampleLCD\) is in \protect\cref{fig:algo:samplelcd}, the algorithm \(\SampleCode\) is in \protect\cref{fig:algo:samplecode}, the algorithm \(\Hull\) is in \protect\cref{fig:algo:hull}.}
  \label{fig:algo:sampletotal}
\end{figure}

\begin{theorem}\label{thm:algo:sampletotal}
  Assume the notation of \cref{def:algo:sampletotal}. The algorithm \(\SampleTotal\) runs in time \(\poly[n,k,q]\), it succeeds with probability \(1 - \negl[k]\) unless \(d k = \frac{n}{2}\), \(h=k\) and \(\chi_p(-1)^{d k} \neq 1\) in which case it always aborts, and samples totally \(h\)-hollow \([n,k]\)-linear codes over \(\FF_q\) with Euclidean square class \(\chi\).
\end{theorem}

\begin{proof}
  The algorithm mimics \(\SampleCode\) but projects to \(\FF_p^n\) to sample total hull vectors. Indeed, we first sample a uniformly random (Euclidean) LCD code of dimension \(k-h\) and square class \(\chi\) over \(\FF_q\) identically to \(\SampleCode\). We then project \(\mat{T}\) to a \(d \cdot (k-h) \times n\) matrix \(\mat{T'}\) over \(\FF_p\) of rank \(t \leq d \cdot (k-h)\). We continue running the code of \(\SampleCode\) but over \(\FF_p\) and with a fixed \(\mat{A}_{(0)} \coloneqq \mat{T'}\) (though \(\mat{A}_{(0)}\) may have a hull over \(\FF_p\)). This produces \(d h\) hull vectors over \(\FF_p\), which we then lift \(d\) at a time using \(\pi^{-1}\) to \(\FF_q^n\) to obtain \(h\) total hull vectors.
  Denote by \(\mat{A}_i\) for \(i=1,\dots,h\) disjoint groups of \(d\) consecutive rows of \(\mat{H}\) and the same for \(i=h+1,\dots,k\) and \(\mat{A}_{(0)}\). Denote also \(\code_i = \code(\mat{A}_i)\). By construction, the codes \(\code_i\) are jointly linearly independent and we have \(\code_i \perp \code_j\) for \(i=1,\dots,h\) and \(j=1,\dots,k\), hence immediately from \cref{thm:orthogonal} and \cref{prop:torthogonality:bases} the algorithm \(\SampleTotal\) is correct. Note that it is possible that the projected matrix \(\mat{A}_{(0)}\) has rank smaller than \(d \cdot (k-h)\), or that it has a hull. However, the algorithm is still correct, since the only vectors that get lifted back to \(\FF_q^n\) are the generated hull vectors which are by \cref{thm:orthogonal} totally orthogonal to \(\mat{T}\) and each other, and linearly independent of \(\mat{T}\). Observe also that the output \(\mat{A}\) generates precisely a totally \(h\)-hollow code, not some larger total hull dimension, since \(\mat{T}\) spans an LCD code and so the Euclidean hull dimension cannot be larger than \(h\). It is possible for an individual Galois hull to have dimension larger than \(h\). Running time and success probability follow from \cref{thm:algo:samplecode}.
\end{proof}

\begin{remark}
  We make no claims in \cref{thm:algo:sampletotal} about the output distribution of \(\SampleTotal\) (conditioned on not aborting), except that its support is a subset of all totally \(h\)-hollow \([n,k]\)-linear codes over \(\FF_q\) with Euclidean square class \(\chi\), though we suspect it is close to the uniform distribution on the whole set.
\end{remark}

We would like to conjecture that the algorithm \(\SampleTotal\) outputs hard instances of \(\PCE\) (and its signed variant) in the following sense. Let \(n, k, h, q = p^d, \chi\) parametrised by \(\secpar\) satisfy the constraints of \cref{thm:algo:sampletotal} with \(h = \Theta(\secpar)\). Then we want that for any \(\ppt\) algorithm \(\adv\) we have that
\[
  \left\lvert \condprob{b = b'}{
    \begin{aligned}
      &\mat{A} \from \SampleTotal(\secpar) \\
      &\mat{B}_0 \from \SampleTotal(\secpar) \\
      &\mat{P} \sample \mathcal{P}_n;\;
      \mat{U} \sample \GL_k(\FF_q)\\
      &\mat{B}_1 \coloneqq \mat{U}\mat{A}\mat{P} \\
      &b \sample \bin;\;
      b' \from \adv^\secpar(\mat{A}, \mat{B})
    \end{aligned}
  } - \frac{1}{2} \right\rvert
  \leq \negl.
\]

However, the price we pay by decreasing the rate also has consequences for security, which we discuss in the next subsection.

\subsection{Limits of PCE over Extension Fields}\label{sec:totalhollow:limits}

As mentioned at the top of this section, the cost we pay is in terms of decreased rate.
That is, up to duality, \(\SampleTotal\) can only sample codes of rate bounded by \(\frac{1}{2 d}\) from above (in the best case where \(\SampleCode\) samples codes of rate \(\frac{1}{2}\)). A recent algorithm by Battagliola, Mora, and Santini~\cite{ITIT:BatMorSan26} shows that codes whose rate is at most inverse-quadratic cannot yield hard permutation equivalence instances, up to some heuristics. Thus the following is the case, showing the cost we pay is indeed non-trivial.

\begin{corollary}\label{cor:attack}
  Let \(p\) odd prime and \(n,k,h,d \in \NN\) with \(k \leq \frac{n}{2 d}\) and \(h \leq k\). Under the assumptions of the attack from~\cite{ITIT:BatMorSan26}, a necessary condition for codes sampled using \(\SampleTotal\) from \cref{def:algo:sampletotal} forming hard (search or distinguish) \(\PCE\) instances is that \(2d \cdot (4d - 1) < n\), that is roughly \(d < \sqrt{\frac{n}{8}}\).
\end{corollary}

\begin{proof}
  Without loss of generality, we may assume \(h=k\), since the algorithm from~\cite{ITIT:BatMorSan26} depends on the rate of the hull. Because the rate of \(\SampleCode\) is bounded by \(\frac{1}{2}\) from above, we must have \(d k \leq \frac{n}{2}\). On the other hand, to thwart the algorithm from~\cite{ITIT:BatMorSan26}, it is required that \(\binom{k+1}{2} > n\) for the coefficient lifted code over \(\FF_q\). Combining these, we obtain that \(4 d k < k \cdot (k+1)\), which implies \(d < \frac{k + 1}{4}\). Multiplying by \(4d\) and rearranging the inequality we get that
  \[
    (2 d)^2 - d < d k \leq \frac{n}{2}
    \qquad
    \text{hence}
    \qquad
    2 d \cdot (4 d - 1) < n.
  \]
  The last inequality is implied by the inequality where we replace \(4 d - 1\) with \(4 d\), giving the estimate \(d < \sqrt{\frac{n}{8}}\).
\end{proof}

What is perhaps worse is that paying the cost of limiting the maximal rate to \(\frac{1}{2 d}\) seems actually necessary when we consider linearisation algorithms for code equivalence~\cite{AC:BCDSK24}. The authors of~\cite{AC:BCDSK24} prove, under assumptions they experimentally verify, that code equivalence can be solved in time \(\bigO{n^{2 \omega}}\) as soon as we have at least \(\floor{\frac{1}{r \cdot (1-r)}} + 1\) samples sharing the same isometry, where \(r = \frac{k}{n}\) is the code rate, and moreover it is even enough to have \(\floor{\frac{1}{2 r \cdot (1-r)}} + 1\) for large enough \(n\) and \(q\), and running time \(\bigO{n^{2 + 2 \omega}}\). The authors of~\cite{upke2} introduced a technique they call Frobenius sample amplification, using the \(d\) automorphisms of \(\FF_{p^d}\) to create \(d\) linearly independent (under a mild assumption they experimentally verify) \(\PCE\) samples from each \(\PCE\) sample over \(\FF_{p^d}\), leading to a linearisation algorithm for \(\PCE\) over extension fields.

\begin{corollary}\label{cor:rate}
  Under the Assumption 1 in~\cite{upke2}, a necessary condition for the hardness of \(\OSPCE\) of rate \(r\) over \(\FF_{p^d}\) for \(p\) prime is \(d < \floor{\frac{1}{r \cdot (1-r)}} + 1\).
\end{corollary}

The Frobenius sample amplification technique from~\cite{upke2} suggests that the rate limitation of \(\frac{1}{2d}\) in \cref{thm:algo:sampletotal} is sufficient for security, but it may not be strictly necessary for large lengths \(n\). Recall the minimum number of required samples to solve \(\PCE\) using linearisation in~\cite{AC:BCDSK24} is \(t = \frac{1}{2 r \cdot (1-r)}\) for codes of rate \(r = \frac{k}{n} \leq \frac{1}{2}\), and that the rate increment for codes of length \(n\) is \(\frac{1}{n}\). The following lemma states that the \(d\) samples we obtain through Frobenius amplification do not break \(\PCE\) if we limit the rate with \(\frac{1}{2d}\) from above as in \cref{thm:algo:sampletotal}, but that we may increase the code dimension slightly and still stay below the sample bound (though it is unclear how to construct such codes algorithmically).

\begin{lemma}\label{lem:tight}
  Let \(t(r) = \frac{1}{2 r \cdot (1-r)}\) for \(r \in (0, 1)\) and \(d \geq 2\). Then \(d < t\left(\frac{1}{2 d}\right)\), and moreover, \(d < t\left(\frac{1}{2 d} + \frac{1}{n}\right)\) for \(n \in \NN\) as soon as \(n > 2 d \cdot (d-1) + 2 d \sqrt{d^2 - 2 d}\).
\end{lemma}

\begin{proof}
  Observe that \(t(r)\) is strictly convex with the only minimum at \(\frac{1}{2}\). Consider the equation \(d = t(r)\), that is \(2 d r \cdot (1-r) = 1\), equivalently \(f_d(r) = r^2 - r + \frac{1}{2 d} = 0\) on \(r \in (0,1)\). Then \(f_d\) is strictly convex and has the only minimum at \(r = \frac{1}{2}\). We have \(d < t(r)\) if and only if \(f_d(r) > 0\).

  For the first claim notice that \(f_d\left( \frac{1}{2 d} \right) = \frac{1}{4 d^2} > 0\). For the second claim, consider
  \begin{align*}
    f_d\left( \frac{1}{2 d} + \frac{1}{n} \right)
    = {\left( \frac{1}{2 d} + \frac{1}{n} \right)}^{2} - \frac{1}{n} &> 0 \\
    \frac{1}{4 d^2} + \frac{1}{d n} + \frac{1}{n^2} - \frac{1}{n} &> 0 \\
    n^2 + 4 d n + 4 d^2 - 4 d^2 n &> 0 \\
    n^2 - 4 d \cdot (d-1) \cdot n + 4 d^2 &> 0
  \end{align*}
  Solving this quadratic inequality for \(n\) we get
  \[
    n > 2 d \cdot (d-1) + 2 d \sqrt{(d-1)^2 - 1}
  \]
  and the conclusion follows. Illustrated in \cref{fig:tight}.
\end{proof}

\begin{figure}[t]
  \centering
  \scalebox{0.9}{
    \begin{tikzpicture}
      \begin{axis}[
          xlabel={\(r\)},
          ylabel={\(f_d(r)\)},
          xmin=0, xmax=0.25,
          ymin=0, ymax=0.20,
          xtick={0,0.05,...,0.25},
          ytick={0,0.05,...,0.20},
          xticklabel style={
            /pgf/number format/fixed,
            /pgf/number format/precision=2
          },
          yticklabel style={
            /pgf/number format/fixed,
            /pgf/number format/precision=2
          },
          scaled x ticks=false,
          scaled y ticks=false,
          xmajorgrids=true,
          ymajorgrids=true,
          grid style=dashed,
          legend pos=north east,
          legend style={font=\small},
        ]
        \addplot[name path=A,domain=0:1,solid] {x*x - x + 1/6};
        \addplot[name path=B,domain=0:1,dashed] {x*x - x + 1/10};
        \addplot[name path=C,domain=0:1,dotted] {x*x - x + 1/14};
        \path [name path=D] (\pgfkeysvalueof{/pgfplots/xmin},0) -- (\pgfkeysvalueof{/pgfplots/xmax},0);
        \addplot[gray!30] fill between [of=A and D, soft clip={domain=1/6:0.21},];
        \addplot[gray!30] fill between [of=B and D, soft clip={domain=1/10:0.21},];
        \addplot[gray!30] fill between [of=C and D, soft clip={domain=1/14:0.21},];
        \node (3) at (1/6,1/36) {\(\bullet\)};
        \node (5) at (1/10,1/100) {\(\bullet\)};
        \node (7) at (1/14,1/196) {\(\bullet\)};
        \legend{\(d=3\),\(d=5\),\(d=7\)};
      \end{axis}
    \end{tikzpicture}
    }
  \caption{The plot of \(f_d(r)\) for \(d \in \set{3,5,7}\)  on \(r \in (0,0.25)\) from the proof of \protect\cref{lem:tight}. The point \(\left( \frac{1}{2 d}, f\left( \frac{1}{2 d} \right) \right)\) is marked on each curve, the gray area is the slack of the rate condition.}
  \label{fig:tight}
\end{figure}
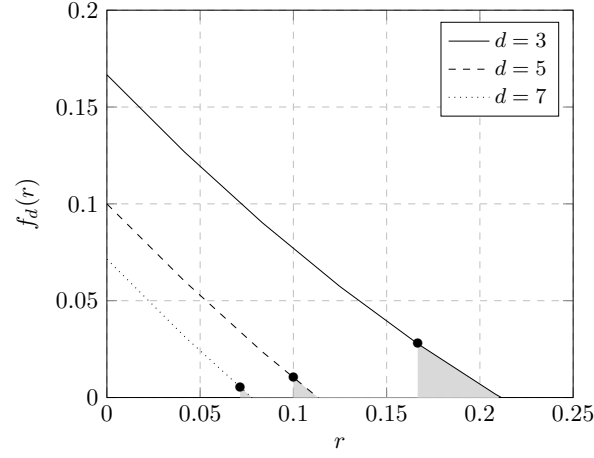

\begin{remark}
  Notice that the condition \(n > 2 d \cdot (4d - 1)\) from \cref{cor:attack} (if we view it as a condition on \(n\) for a fixed \(d\) instead of the other way around) implies the condition \(n > 2 d \cdot (d-1) + 2 d \sqrt{(d-1)^2 - 1}\) in \cref{lem:tight}, thus our reduction in \cref{thm:algo:sampletotal} is not tight for any \(n\) satisfying the necessary condition for security as outlined in \cref{cor:attack}.
\end{remark}

\subsection{Towards Galois Square Classes}\label{sec:total:arith}

It was recently shown in~\cite{localinv} that the hull dimension naturally arises as an invariant for \(\PCE\) when studying linear codes through arithmetic properties of their Gram matrices, along with a previously unrecognised invariant the author calls the \emph{square class} of a code that corresponds to a generalised notion of the discriminant of a quadratic form. It is shown that for a linear code \(\code\), the square class of the determinant of the Gram matrix of any generator of \(\code/\hull(\code)\) is invariant under (signed) permutation equivalence. In odd characteristic, this divides codes of every hull dimension \(h < k\) into two classes of roughly the same size~\cite{ITIT:LiShiLin25}, while self-orthogonal codes with \(h=k\) lie in a single class.

Given that the notion of the hull of a linear code generalises to Galois geometries, it seems natural that the same would be the case for square classes. In this section, we attempt to generalise square classes to Galois geometries, but leave the formal treatment of this subject analogous to~\cite{ITIT:LiShiLin25,localinv} for future work.

\begin{proposition}\label{prop:galois:sqclass}
  For every \(i \in \ZZ_d\) and \(a \in \FF_q^\star\) for \(q = p^d\) odd prime power, \(a \phi_i(a)\) is a quadratic residue.
\end{proposition}

\begin{proof}
  We have \(a \phi_i(a) = a^{p^i + 1}\) with \(p^i + 1\) even.
\end{proof}

The following statement generalises the square class invariant~\cite{localinv} to Galois geometries for Galois LCD codes. As in the Euclidean case, a linear code \(\code\) is Galois LCD~\cite{FFA:LiuFanLiu18} with respect to \(\phi_i\) if \(\code \cap \code^{\perp,\phi_i} = \set{\vec{0}}\).

\begin{proposition}\label{prop:galois:sqclassmat}
  Let \(\phi_i \in \Gal(\FF_q/\FF_p)\) and let \(\mat{A},\mat{B} \in \FF_q^{k \times n}\) generate signed permutation equivalent Galois LCD codes with respect to \(\phi_i\). Then \(\chi_q(\det(\mat{A}\phi_i(\mat{A}^\transpose))) = \chi_q(\det(\mat{B}\phi_i(\mat{B}^\transpose)))\).
\end{proposition}

\begin{proof}
  There exist \(\mat{O} \in \Ort_n(\ZZ)\) and \(\mat{U} \in \GL_k(\FF_q)\) such that \(\mat{B} = \mat{U}\mat{A}\mat{O}\), hence it holds for Gram matrices that \(\det(\mat{B}\phi_i(\mat{B}^\transpose)) = \det(\mat{U})\phi_i(\det(\mat{U})) \det(\mat{A}\phi_i(\mat{A}^\transpose))\) since \(\phi_i(\mat{O}) = \mat{O}\). Then use \cref{prop:galois:sqclass} for \(a = \det(\mat{U})\).
\end{proof}

To define the invariant for non-LCD codes (where the determinant of the skewed Gram matrix is \(0\)), codes must be `dehulled', analogously to~\cite{localinv}. However, a Galois LCD component is in general not well-defined, because the inner product need not be symmetric. Indeed, an LCD component with respect to the Euclidean inner product has the property that it remains an LCD component even if we add to its rows linear combinations of hull vectors, which is not the case here. If \(\vec{h} \in \hull_{\phi_i}(\code)\) and \(\vec{t}\) is in the span of some matrix \(\mat{T}\) that modulo \(\hull_{\phi_i}(\code)\) generates \(\code/\hull_{\phi_i}(\code)\), we have that \(\inner{\vec{t} + \vec{h}}{\vec{t} + \vec{h}}_{\phi_i} = \inner{\vec{t}}{\vec{t}}_{\phi_i} + \inner{\vec{h}}{\vec{t}}_{\phi_i}\), and the additional term may flip the square class of \(\mat{T}\) as defined in \cref{prop:galois:sqclassmat}. However, this strategy does work for \(i \in \ZZ_d\) with \(\hull_{\phi_i}(\code) = \hull_{\phi_{-i}}(\code)\). We focus on this symmetric case for the rest of this section and leave investigating the general case for future work.

\begin{proposition}[Galois dehulling]\label{prop:galois:dehull}
  Let \(\mat{A}\) generate an \([n,k]\)-linear code \(\code\) over \(\FF_q\) for \(q=p^d\) odd prime power, and suppose that for some \(i \in \ZZ_d\) we have that \(\hull_{\phi_i}(\code(\mat{A})) = \hull_{\phi_{-i}}(\code(\mat{A}))\) and has dimension \(h < k\). The algorithm in \cref{fig:galois:dehull} runs in time \(\bigO{n^\omega}\) and \emph{dehulls} \(\mat{A}\) with respect to \(\phi_i\), that is, finds a matrix \(\mat{T} \in \FF_q^{(k-h) \times n}\) that generates a Galois LCD code such that \(\code = \hull_{\phi_i}(\code) \oplus \code(\mat{T})\).
\end{proposition}

\begin{proof}
  Using linear algebra, find representatives in \(\code\) for the basis of the quotient \((k-h)\)-dimensional \(\FF_q\)-vector space \(\code/\hull_{\phi_i}(\code)\) and store them as \(\mat{T}\). Every \(\vec{x} \in \hull_{\phi_i}(\code(\mat{T}))\) is (left and right) \(\phi_i\)-orthogonal to all of \(\code(\mat{T})\) and to \(\hull_{\phi_i}(\code)\), hence \(\vec{x}\) is (left and right) \(\phi_i\)-orthogonal to all of \(\code\). Thus \(\vec{x} \in \hull_{\phi_i}(\code)\), and \(\hull_{\phi_i}(\code) \cap \code(\mat{T}) = \set{\vec{0}}\) implies \(\vec{x} = \vec{0}\).
\end{proof}

\begin{figure}[t]
  \centering
  \begin{pchstack}[boxed]
    \procedure[mode=text]{$\DehullGalois(\mat{A}, i)$}{
      \(\mat{H} \coloneqq \HullGalois(\mat{A}, i)\);\;
      \(\mat{V} \coloneqq\) basis of \(\code(\mat{A}) / \code(\mat{H})\) \\
      \(\mat{T} \coloneqq\) lift \(\mat{V}\) to \(\code(\mat{A})\);\;
      \pcreturn \((\mat{T}, \mat{H})\)
    }
  \end{pchstack}
  \caption[The algorithm \(\DehullGalois\).]{The algorithm \(\DehullGalois\). We may also return a hull generator matrix for convenience. The algorithm \(\HullGalois\) is in \protect\cref{fig:galois:hull}.}
  \label{fig:galois:dehull}
\end{figure}

\begin{definition}[Galois square class]\label{def:galois:sqclass}
  Let \(\code\) be an \([n,k]\)-linear code over \(\FF_q\) for \(q = p^d\) odd prime power, and suppose that for some \(i \in \ZZ_d\) we have that \(\hull_{\phi_i}(\code) = \hull_{\phi_{-i}}(\code)\) and has dimension \(h < k\). The \emph{Galois square class} of \(\code\) with respect to \(\phi_i\) is \(\chi_{\code,\phi_i} = \chi_q(\det(\mat{T}\phi_i(\mat{T}^\transpose))) \in \set{\pm 1}\) for any Galois LCD component \(\mat{T}\) with respect to \(\phi_i\). When \(h = k\), all codes have square class \(1\).
\end{definition}

\begin{proposition}\label{prop:galois:sqclass:comput}
  Assume the notation of \cref{def:galois:sqclass}. There is a deterministic algorithm that, given any generator matrix \(\mat{A}\) for \(\code\) and admissible \(i \in \ZZ_d\), computes \(\chi_{\code,\phi_i}\) in time \(\bigO{n^\omega + \log^2(q)}\).
\end{proposition}

\begin{proof}
  If \(\code\) is \(\phi_i\)-self-orthogonal, then output \(\chi_{\code,\phi_i} = 1\). Otherwise, dehulling and computing the determinant each take \(\bigO{n^\omega}\) by \cref{prop:galois:dehull}, and applying \(\chi_q\) takes \(\bigO{\log^2(q)}\) by~\cite[1.4.12]{Cohen1993}.
\end{proof}

\begin{corollary}\label{cor:invariant:galois}
  Assume the notation of \cref{def:galois:sqclass}.
  The function \(\chi^{\phi_i} \colon \FF_q^{k \times n} \to \set{\pm 1}\) that maps \(\mat{A} \mapsto \chi_{\code(\mat{A}),\phi_i}\) for admissible \(i \in \ZZ_d\) is an efficiently computable invariant for (signed) permutation equivalence.
\end{corollary}

\begin{proof}
  Invariance is immediate from \cref{prop:galois:sqclassmat,prop:galois:dehull} and \(\phi(\mat{O}) = \mat{O}\) for \(\mat{O} \in \Ort_n(\ZZ)\) a signed permutation, efficiency from \cref{prop:galois:sqclass:comput}.
\end{proof}

Since we have assumed \(\hull_{\phi_i}(\code) = \hull_{\phi_{-i}}(\code)\), it is easy to see that Galois square classes share a similar symmetry.

\begin{proposition}\label{prop:galois:sqclass:sym}
  For every linear code \(\code\) over \(\FF_q\) for \(q = p^d\) odd prime power we have that if \(i \in \ZZ_d\) is such that \(\hull_{\phi_i}(\code) = \hull_{\phi_{-i}}(\code)\), then \(\chi_{\code,\phi_i} = \chi_{\code,\phi_{-i}}\).
\end{proposition}

\begin{proof}
  Without loss of generality, we may assume \(\code(\mat{A})\) is Galois LCD with respect to \(\phi_i\), since a Galois LCD component is well-defined (\cref{prop:galois:dehull}). Then \(\det(\mat{A}\phi_i(\mat{A}^\transpose)) = \phi_i(\det(\mat{A}\phi_{-i}(\mat{A}^\transpose)))\). By \cref{prop:galois:sqclass}, \(a \in \FF_q\) is a quadratic residue if and only if \(\phi_i(a)\) is a quadratic residue, hence \(\chi_q(\det(\mat{A}\phi_i(\mat{A}^\transpose))) = \chi_q(\det(\mat{A}\phi_{-i}(\mat{A}^\transpose)))\).
\end{proof}

In particular, the Galois square class should work in the Hermitian case of \(d\) even and \(i = \frac{d}{2}\); however, the invariant is constant in this case.

\begin{proposition}\label{prop:hermitian:sqclass}
  Assume the notation of \cref{def:galois:sqclass}, and additionally suppose \(d\) is even and fix \(i = \frac{d}{2}\). Then for every linear code \(\code\) we have \(\chi_{\code,\phi_i} = 1\).
\end{proposition}

\begin{proof}
  Let \(\mat{A}\) generate \(\code\) and let \(\mat{M} = \mat{A} \phi_i(\mat{A}^\transpose)\). Since the form is Hermitian, we have \(\mat{M} = \phi_i(\mat{M}^\transpose)\). Hence \(\det(\mat{M}) = \det(\mat{M}^\transpose) = \det(\phi_i(\mat{M})) = \phi_i(\det(\mat{M}))\), so \(\det(\mat{M})\) lies in the subfield \(\FF_{p^\frac{d}{2}}\) fixed by \(\phi_i\). Now observe that every element of \(\FF_{p^\frac{d}{2}}\) is a quadratic residue in \(\FF_{p^d}\); if \(a \in \FF_{p^\frac{d}{2}}\) is embedded into \(\FF_{p^d}\), then \(\chi_{p^d}(a) = a^\frac{p^d-1}{2} = a^{\left(p^\frac{d}{2}-1\right)\cdot\frac{p^\frac{d}{2}+1}{2}} = 1\).
\end{proof}

Using either \(\SampleCode\) in \cref{def:algo:samplecode} or \(\SampleTotal\) in \cref{def:algo:sampletotal}, it is easy to generate linear codes of a prescribed (total) hull dimension whose square classes match, yet some of their Galois square classes differ, meaning it is easy to decide that such codes are \emph{not} (signed) permutation equivalent. The advantage we wished for above when proposing \(\SampleTotal\), which like \(\SampleCode\) only prescribes the Euclidean square class, cannot be negligible, though this does not affect permutation equivalence over finite extension field of degree \(2\) by \cref{prop:hermitian:sqclass}. Investigating the exact behaviour of Galois square classes in relation to (signed) permutation equivalence analogously to~\cite{localinv} would require adapting the mass formulae in~\cite{ITIT:LiShiLin25} to Galois geometries, which we consider future work.

\section{Implementation}\label{sec:experiments}

Our SageMath~\cite{sagemath} implementation of the algorithms in \cref{sec:hollow:sample} is provided in \ifsupplementary{\textattachfile{hollow.py}{\texttt{hollow.py}}}{\texttt{hollow.py} as supplementary material}. We also provide in \ifsupplementary{\textattachfile{galois.py}{\texttt{galois.py}}}{\texttt{galois.py}} a proof-of-concept implementation of the sampler from \cref{fig:algo:sampletotal}, along with all the auxiliary functions for dealing with different geometries. Both depend on \ifsupplementary{\textattachfile{lincodes.py}{\texttt{lincodes.py}}}{\texttt{lincodes.py}} and \ifsupplementary{\textattachfile{util.py}{\texttt{util.py}}}{\texttt{util.py}}, which provide functions for dealing with linear codes (including dehulling and computing the type) and general utility functions, respectively.

\bibliographystyle{alpha-doi} 
\IfFileExists{cited.bib}{
\bibliography{cited}
}{
\bibliography{abbrev0,crypto,local,sw,extra}
}

\appendix

\section*{Acknowledgment}

We thank Paolo Santini for pointing us towards~\cite[Prop.~8]{ISIT:BarOtmSae19}.

\end{document}
